\documentclass[12pt,a4paper]{article}

\let\ifpdf\undefined
\usepackage[dvips,pdftex]{graphicx}
\usepackage{epsfig,multicol,bbm}
\usepackage{amsfonts,amsmath,amssymb,amsthm,latexsym}
\usepackage[all]{xy}
\newcommand\fverb{\setbox\fverbbox=\hbox\bgroup\verb}
\newcommand\fverbdo{\egroup\medskip\noindent%
			\fbox{\unhbox\fverbbox}\ }
\newcommand\fverbit{\egroup\item[\fbox{\unhbox\fverbbox}]}
\newbox\fverbbox
\newtheorem{proposition}{Proposition}[section]
\newtheorem{lemma}[proposition]{Lemma}
\newtheorem{corollary}[proposition]{Corollary}
\newtheorem{theorem}[proposition]{Theorem}

\theoremstyle{definition}
\newtheorem{definition}[proposition]{Definition}
\newtheorem{example}[proposition]{Example}

\theoremstyle{remark}
\newtheorem{remark}[proposition]{Remark}

\newcommand{\lelabel}[1]{\label{le:#1}}
\newcommand{\leref}[1]{Lemma~\ref{le:#1}}
\newcommand{\prlabel}[1]{\label{pr:#1}}
\newcommand{\prref}[1]{Proposition~\ref{pr:#1}}

\newcommand{\relabel}[1]{\label{re:#1}}
\newcommand{\reref}[1]{Remark~\ref{re:#1}}
\newcommand{\exlabel}[1]{\label{ex:#1}}
\newcommand{\exref}[1]{Example~\ref{ex:#1}}
\newcommand{\delabel}[1]{\label{de:#1}}
\newcommand{\deref}[1]{Definition~\ref{de:#1}}
\newcommand{\eqlabel}[1]{\label{eq:#1}}
\newcommand{\equref}[1]{(\ref{eq:#1})}

\title{The notion of abstract manifold: a pedagogical approach}

\author{Konstantinos Kanakoglou \\
School of Mathematics, \\
Aristotle University of Thessaloniki (AUTh), \\
54124, \\
Thessaloniki,  \\
Greece  \\  \\
\small{kanakoglou@hotmail.com, \\ kanakoglou@ifm.umich.mx}}

\abstract{A self-contained introduction is presented of the notion of the (abstract) differentiable manifold and its tangent vector fields. The way in which elementary topological ideas stimulated the passage from Euclidean (vector) spaces and linear maps to abstract spaces (manifolds) and diffeomorphisms is emphasized. Necessary topological ideas are introduced at the beginning in order to keep the text as self-contained as possible. Connectedness is presupposed in the definition of the manifold. Definitions and statements are laid rigorously, lots of examples and figures are scattered to develop the intuitive understanding and exercises of various degree of difficulty are given in order to stimulate the pedagogical character of the manuscript. The text can be used for self-study or as part of the lecture notes of an advanced undergraduate or beginning graduate course, for students of mathematics, physics or engineering. 

\textbf{keywords:} topological manifolds, differentiable manifolds, tangent space, tangent bundle, tangent vector field}

\begin{document}

\maketitle

\acknowledgments{The manuscript grew from the teaching of a postgraduate course of Differential Geometry, during the second semester of the academic period 2009-2010, at physics students of the Institute of Physics and Mathematics (IFM) of the University of Michoacan (UMSNH) at Morelia, Michoacan, Mexico. The author would like to kindly thank the students attending the course, whose inspiring  interest on the foundations of modern geometry, stimulated the writing of a set of lecture notes. It was from the first part, of the initial form of these lecture notes, that this manuscript evolved. The author would also like to express his gratitude to Prof. Alfredo Herrera-Aguilar. During the preparation of the present manuscript KK was supported by a position of Visiting Professor (2009-2010) at IFM, UMSNH and by a postdoctoral fellowship (2011-2012) granted from the Research Committee of the Aristotle University of Thessaloniki (AUTh).  \\ }

%\textbf{MSC2010:} 5301, 9701, 5101, 97U30, 97D80

%\textbf{keywords:} topological manifolds, differentiable manifolds, tangent space, tangent bundle, tangent vector field

\newpage

\section{General Topology: a Review}

Topology began with the investigation of certain questions in geometry. Euler's $1736$ paper on the Seven Bridges of K\"{o}nigsberg is regarded as one of the first academic treatises in modern topology. Towards the end of the $19^{th}$ century, a distinct discipline developed, which was referred to in Latin as the \emph{Geometria Situs} ("geometry of place") or \emph{Analysis Situs}, and which later acquired the modern name of topology.
The concept of topological space grew out of the study and the abstraction of the properties of the real line and Euclidean space and the study of continuous functions and the notion of ``\emph{closedness}" in these spaces. The definition of topological space that is now standard -and presented here- was for a long time being formulated. Various mathematicians -between them, Frechet, Hausdorff, Kuratowski and others as well- proposed different definitions over a period of years ranging from the late $19^{th}$ to the first half of the $20^{th}$ century and it took quite a while before we settled on the one that seemed most suitable.

\subsection{General notions}

We start by recalling the fundamental notions of general Topology.
\begin{definition}[\textbf{Topological Space}] \delabel{TopologicalSpace}
Let $X$ be an (arbitrary) set and consider a collection $\tau = \{\omega / \omega \subseteq X \}$ of subsets of $X$. Any
$\omega \in \tau$ will be called an \textbf{open set}. The collection $\tau$ will be said to define a \textbf{topology on} $\mathbf{X}$ if
the following conditions are satisfied
\begin{enumerate}
\item $\emptyset, X \in \tau$
\item For any collection $\{\omega_{i}\}_{i \in I}$ with $\omega_{i} \in \tau$ $\forall i \in I$  ($I$ an arbitrary index set), we have
$\bigcup_{i \in I}\omega_{i} \in \tau$
\item For any finite collection $\omega_{1}, \omega_{2}, ..., \omega_{\kappa} \in \tau$  we have $\bigcap_{i = 1}^{\kappa} \omega_{i} \in \tau$
\end{enumerate}
$\mathbf{(X, \tau)}$ will be called a \textbf{topological space}.
\end{definition}
The notion of \textbf{closed set} is defined as a set whose complement is open: Thus, if $\omega \in \tau$ is an open set then its complement $X \smallsetminus \omega$ is
closed and conversely if $U \subseteq X$ is a closed subset of $X$ then $X \smallsetminus U$ is open. Note that by definition both $X$ and $\emptyset$ are open and
closed at the same time (thus the notions of open and closed sets are not mutually exclusive!).

We will now present a series of examples, starting from a couple of somewhat extreme cases:
\begin{example}
Let $X$ be a set and consider the collection of subsets $\tau_{0} = \{ \emptyset, X \}$. It is topology (verify!). We will call this as the \textbf{trivial}
or \textbf{minimal} topology on $X$
\end{example}
\begin{example}
Let $X$ be a set and consider the collection of subsets given by $\tau_{1} = \{ all \ possible \ subsets \ of \ X \}$. It is a topology (verify!). We will call
this as the \textbf{discrete} or \textbf{maximal} topology on $X$
\end{example}
\begin{example}[\textsf{Topologies on a finite set}]
Let $X = \{ a, b, c \}$ be a three-element set. There are many possible topologies on $X$, some of which are indicated schematically in the first nine drawings of the following figure.
\begin{center}
\includegraphics[angle=270, scale=0.4]{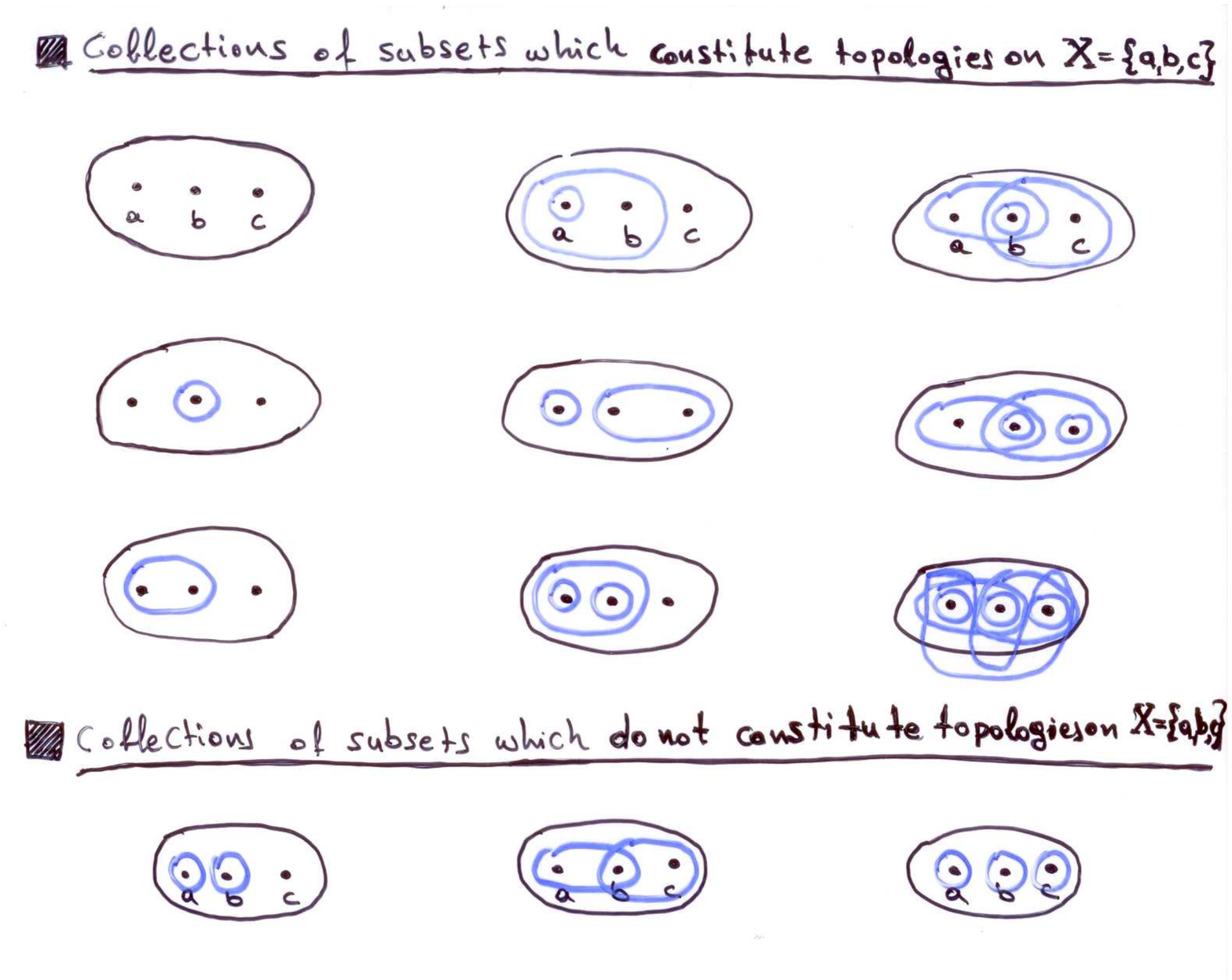}
\end{center}
But -as is schematically indicated in the last three drawings of the previous figure- not every collection of subsets of $X$ constitutes a topology on $X$
\end{example}
\begin{example}[\textsf{Metric Topological spaces}]
Any metric \footnote{We recall here the definition of a metric on an arbitrary set $X$: \\ \textbf{Def.:} Given any set $X$, a \textbf{metric on} $\mathbf{X}$ is a \underline{real},
\underline{finite}, \underline{non-negative} function $\rho : X \times X \rightarrow \mathbb{R}$ such that the following axioms are satisfied $\forall x,y,z \in X$ \\
$\bullet$ \emph{identity}: $\rho(x,y) = 0$  $\ \Leftrightarrow$ $\ x=y$ \\
$\bullet$ \emph{symmetry}:  $\rho(x,y)=\rho(y,x)$  \\
$\bullet$ \emph{triangle inequality}:  $\rho(x,y) \leq \rho(x,z) + \rho(z,y)$   \\
The pair $\mathbf{(X, \rho)}$ will be called a \textbf{metric space}.
} space is a topological space with topology induced by the metric. In other words $\tau$ is equal to the collection of all subsets of $X$ which are open in the metric space sense.
This is called the \textbf{metric topology} on $X$. Let us describe this in some more detail: In any metric space $(X, \rho)$ we can define $\forall x_{0} \in X$ the following subsets
\begin{equation} \eqlabel{balls-spheres}
\begin{array}{c}
B(x_{0}, r) = \{ x \in X \diagup \rho(x, x_{0}) < r \} \rightsquigarrow \mathbf{Open \ ball} \ (\underline{centered \ at \ \ \mathbf{x_{0}}}, \ with \ \underline{radius \ \mathbf{r}})  \\
\widetilde{B}(x_{0}, r) = \{ x \in X \diagup \rho(x, x_{0}) \leq r \} \rightsquigarrow \mathbf{Closed \ ball} \ (\underline{centered \ at \ \mathbf{x_{0}}}, \ with \ \underline{radius \ \mathbf{r}})  \\
S(x_{0}, r) = \{ x \in X \diagup \rho(x, x_{0}) = r \} \rightsquigarrow \mathbf{Sphere} \ (\underline{centered \ at \ \mathbf{x_{0}}}, \ with \ \underline{radius \ \mathbf{r}})
\end{array}
\end{equation}
A subset of a metric space will be called \emph{open} -by defin.- if it contains an open ball about each of its points. It can be shown, that in a metric space a set is open if and only if it can be expressed as a union of open balls.
\end{example}
\begin{example}[\textsf{Euclidean Topology on $\mathbb{R}^{n}$}]
The typical example of metric topological spaces is of course the $\mathbb{R}^{n} = \{ \overline{x} = (x_{1}, x_{2}, ..., x_{n}) \diagup x_{i} \in \mathbb{R}, \ \forall \ i = 1,2, ..., n \}$ with metric given by
\begin{equation}
\rho(\overline{x},\overline{y}) = \sqrt{\sum_{k=1}^{n}(x_{k}-y_{x})^{2}} = |\overline{x}-\overline{y}|
\end{equation}
The metric topology induced on $\mathbb{R}^{n}$ by the above metric, is frequently called in the bibliography as \textbf{Euclidean topology}.
\end{example}
\textbf{\underline{Exercises:}}  \\
$\mathbf{1.}$ Let $X = \{x_{1}. x_{2}, ..., x_{n}\}$ be a finite set with $n$-elements. Denote by $\tau(X)$ the set of all of its subsets. This is called the \emph{power set} of $X$. How many elements does $\tau(X)$ have ? How many possible collections of subsets can we choose from it ? Does each one of these collections constitutes a topology ?
\begin{definition}[\textbf{Base} of a Topology]
A collection $\mathcal{B} = \{ V \} \subseteq \tau$ of open sets of a topological space $(X, \tau)$ will be called a \textbf{base} of the topology if for any open set
$w$ and any point $x \in w$, there exists some $V \in \mathcal{B}$ such that: $x \in V \subset w$
\end{definition}
The following lemma provides an equivalent description of the notion of the base:
\begin{lemma}
Let $(X, \tau)$ be a topological space and let $\mathcal{B} = \{ V \} \subseteq \tau$ a collection of open sets. $\mathcal{B}$ is a base of the topology if and only if every
non-empty open subset of $X$ can be represented as a union of a subfamily of $\mathcal{B}$.
\end{lemma}
\begin{example}
If $X$ is any set, the collection of all one point subsets of $X$ is a basis for the discrete topology on $X$
\end{example}
\begin{example}
The collection of all open balls (see \equref{balls-spheres}) about every point $\overline{x} \in \mathbb{R}^{n}$ forms a base for the Euclidean topology of $\mathbb{R}^{n}$
\end{example}
\begin{example}[\textsf{Bases in $\mathbb{R}^{2}$ Euclidean topology}]
We will illustrate two different bases for the $\mathbb{R}^{2}$ metric topology, one made of ``\emph{balls}" and another made of ``\emph{boxes}". Lets describe the situation visually:
\begin{center}
\includegraphics[scale=0.35]{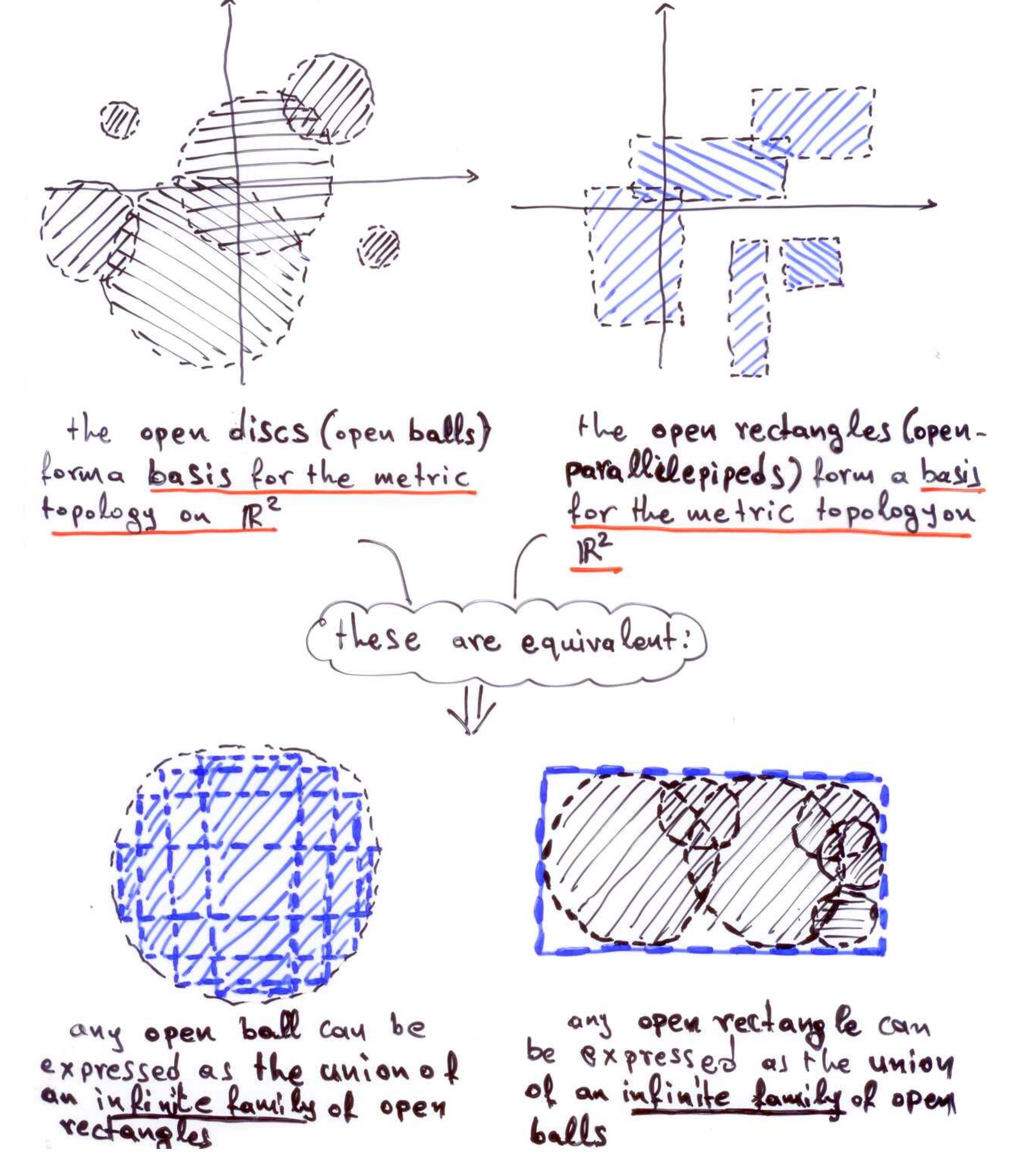}
\end{center}
\end{example}
\begin{definition}
Let $(X, \tau)$ be a topological space. A collection $\mathcal{S} \subset \tau$ of open sets is called a \textbf{subbase} for the topology if the family of all finite intersections
$\cap_{i=1}^{\kappa}U_{i}$ where $U_{i} \in \mathcal{S}$ for all $i=1, 2, ..., \kappa$ forms a base of the topology $\tau$
\end{definition}
Let us proceed to one more visual example on the notion of subbase:
\begin{example}[\textsf{Subbase in $\mathbb{R}^{2}$ Euclidean topology}]
Intersections between an horizontal and a vertical open strip produce an open box (rectangle) with its sides parallel to the coordinate axes. The collection of all such open boxes constitutes a basis for the $\mathbb{R}^{2}$ Euclidean topology. The situation is visualized in the following figure:
\begin{center}
\includegraphics[angle=270, scale=0.35]{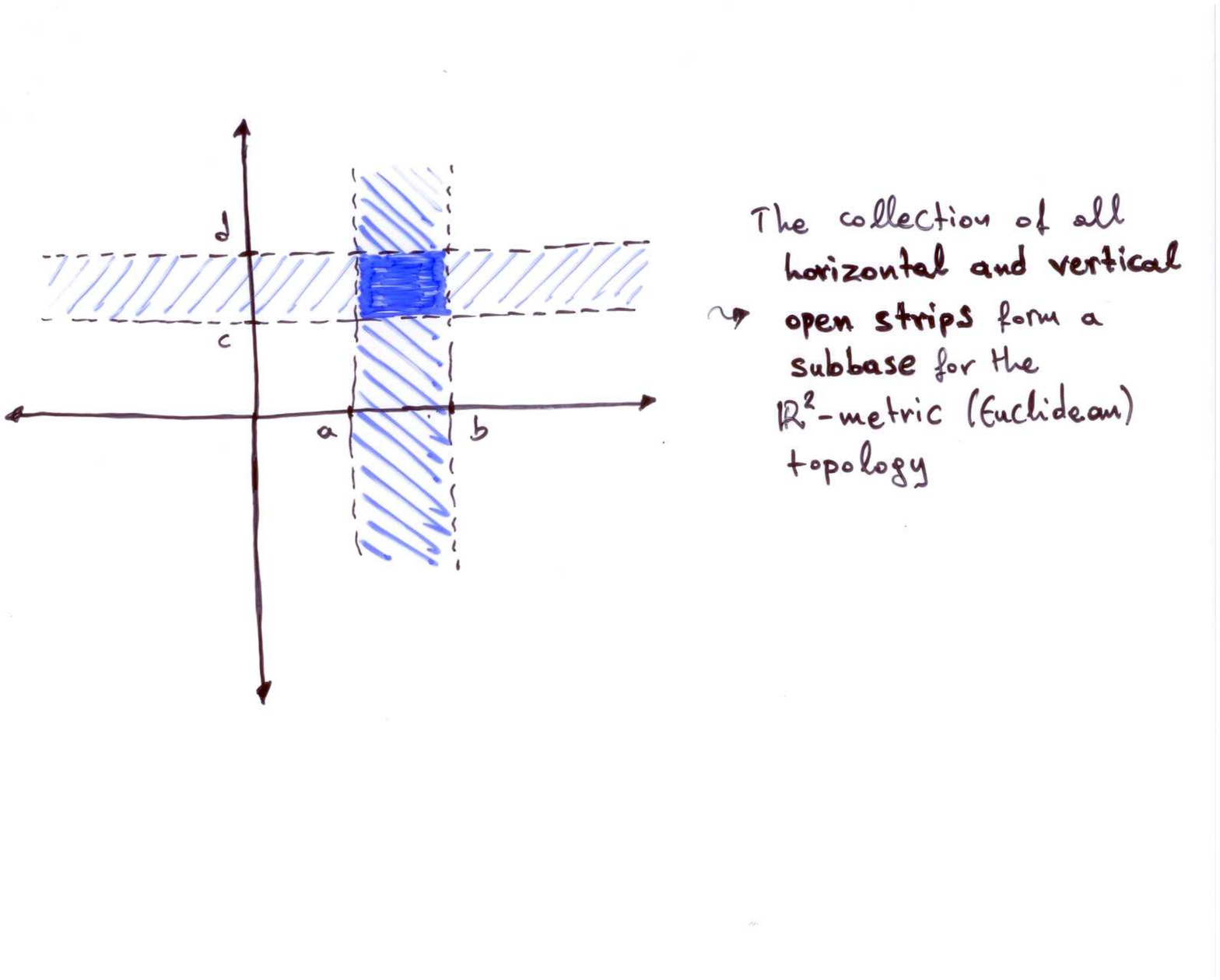}
\end{center}
\end{example}
Finally, before closing this paragraph, let us present the definitions of some notions which will prove valuable in the sequel
\begin{definition}
Let $(X, \tau)$ be a topological space and let $A \subseteq X$ be any subset.  \\
%\begin{itemize}
$\bullet$ The \textbf{interior} $\mathbf{{\AA}}$ of $A$ is defined to be the union of all open sets contained in $A$: $\mathbf{{\AA} = \bigcup_{i \in I}\omega_{i}}$ where $\omega_{i}:open$ and $\omega_{i} \subseteq A, \ \forall \ i \in I$  \\
$\bullet$ The \textbf{closure} $\mathbf{\overline{A}}$ of $A$ is defined to be the intersection of all closed sets containing $A$: $\mathbf{\overline{A} = \bigcap_{i \in J}A_{i}}$ where $A_{i}:closed$ and $A \subseteq A_{i}, \ \forall \ i \in J$    \\
$\bullet$ The \textbf{boundary} $\mathbf{\partial A}$ of $A$ is defined to be $\mathbf{\partial A = \overline{A} \setminus {\AA}}$   \\
$\bullet$ Let $A \subset B \subset X$. $A$ will be called \textbf{dense} in $B$, if $\mathbf{\overline{A} = B}$   \\
$\bullet$ A subset $\mathbf{V_{x}}$ of $X$ containing the point $x \in X$, will be called a \textbf{neighborhood of} $\mathbf{x}$ if it contains an open subset containing $x$. Formally: $V_{x} \subseteq X$ will be called a neighborhood of $x \in X$ if $\exists \ \omega \in \tau$ such that $x \in \omega \subseteq V_{x}$
%\end{itemize}
\end{definition}
\begin{lemma}
We record a couple of properties resulting from the previous definition:  \\
%\begin{itemize}
$\bullet$ ${\AA}$ is open and in fact is the largest open set contained in $A$   \\
$\bullet$ $\overline{A}$ is closed and in fact is the smallest closed set containing $A$   \\
$\bullet$ $A \subseteq X$ will be open if and only if $A = {\AA}$    \\
$\bullet$ $A \subseteq X$ will be closed if and only if $A = \overline{A}$   \\
$\bullet$ A subset $A \subseteq X$ will be open if and only if $A$ is a neighborhood of all of its points
%\end{itemize}
\end{lemma}

\subsection{Continuity-Homeomorphisms}
We give now an introduction to the notion of continuity:
\begin{definition}[\textbf{Continuity a la Cauchy}] \delabel{continatpointCauchy}
Let $(X, \tau_{X})$ and $(Y, \tau_{Y})$ be topological spaces and $f:X \rightarrow Y$ a map between them. We will say that $\mathbf{f}$ is \textbf{continuous at the point}
$\mathbf{x_{0} \in X}$ if for any neighborhood $V_{f(x_{0})}$ of the point $f(x_{0}) \in Y$ there is a neighborhood $U_{x_{0}}$ of the point $x_{0} \in X$, such that
$f(U_{x_{0}}) \subset V_{f(x_{0})}$. We will say that $\mathbf{f}$ is a \textbf{continuous function} if it is continuous at any $x \in X$.
\end{definition}
Cauchy's conception of continuity is visualized in the following
\begin{center}
\includegraphics[scale=0.5]{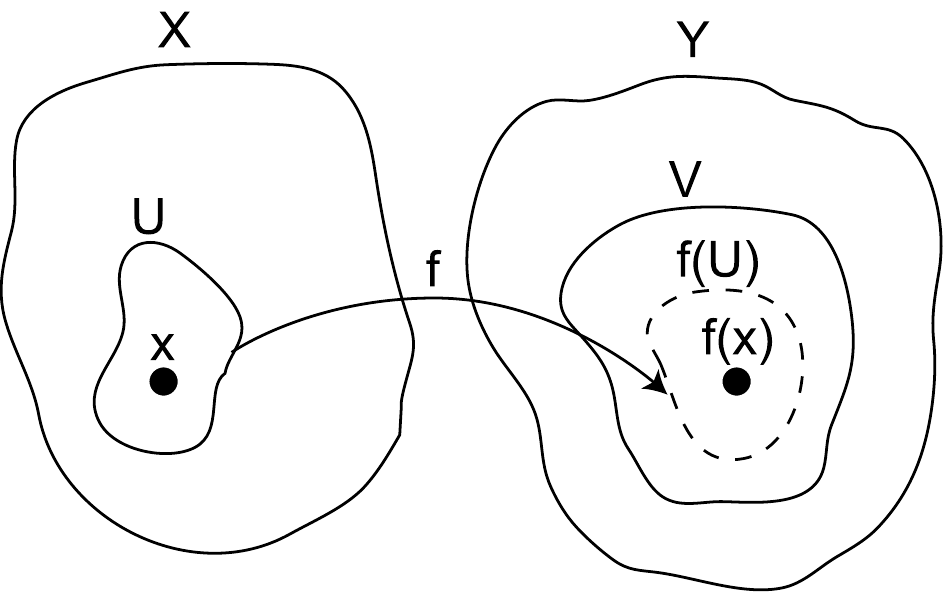}
\end{center}
The reader should notice at this point that this form of the definition of continuity is very close to the
familiar $(\varepsilon(\delta), \delta)$-definition of elementary calculus. Now we can prove the following:
\begin{proposition}[local continuity]
The function $f:X \rightarrow Y$ is continuous at the point $x_{0} \in X$ if for any neighborhood $V_{f(x_{0})}$ of the point $f(x_{0}) \in Y$ the set
$f^{-1}(V_{f(x_{0})})$ is a neighborhood of $x$ in $X$
\end{proposition}
\begin{proposition}[global continuity]
The function $f:X \rightarrow Y$ is continuous if and only if one of the following equivalent conditions is satisfied:
\begin{enumerate}
\item The inverse image of any open set is an open set
\item The inverse image of any closed set is a closed set
\end{enumerate}
\end{proposition}
Now we have the following
\begin{definition}
A mapping $f:X \rightarrow Y$ is said to be \textbf{open} (resp. \textbf{closed}) if the image of any open (resp. closed) set in $X$ is an open (resp. closed) set in $Y$
\end{definition}
We have to remark here that an open map does not necessarily have to be closed and that a continuous map does not necessarily have to be
\footnote{Take for example the continuous function $f:\mathbb{R} \rightarrow \mathbb{R}$ given by $f(x)=x^{2}$ and note that the open interval $(-1,1)$ maps onto the
interval $[0,1)$ which is not open} neither open nor closed. Thus, the notions ``continuous", ``open" and ``closed", are in principle independent to each other.  \\
\textbf{\underline{Exercises:}}   \\
Using the above def., and considering a bijective function $f$, show the following:   \\
$\mathbf{2.}$ $f$: continuous $\Leftrightarrow$ $f^{-1}$: open, \ \ \ \ \
$\mathbf{3.}$ $f$: continuous $\Leftrightarrow$ $f^{-1}$: closed    \\
$\mathbf{4.}$ $f$: open $\Leftrightarrow$ $f^{-1}$: continuous, \ \ \ \ \ \
$\mathbf{5.}$ $f$: closed $\Leftrightarrow$ $f^{-1}$: continuous   \\
$\mathbf{6.}$ If $(X, \tau_{1})$ is a topological space equipped with the discrete (or: maximal topology $\tau_{1}$) and $(Y, \tau_{Y})$ is any topological space, prove that
any map $f:X \rightarrow Y$ is continuous   \\
$\mathbf{7.}$  If $(X, \tau_{X})$ is any topological space and $(Y, \tau_{0})$ is a topological space equipped with the trivial (or: minimal topology $\tau_{0}$), prove that
any map $f:X \rightarrow Y$ is continuous   \\
$\mathbf{8.}$ Consider two different topologies $ \tau_{1} $, $ \tau_{2} $ on the same set $X$. Prove that the identity map \footnote{It is easy to see that the term ``identity" map for $id_{X}:id_{X}(x)=x, \ \forall x \in X$ is in fact used only in the set-theoretic sense but not in the topological sense} $id_{X}:(X, \tau_{2}) \rightarrow (X, \tau_{1})$ is a continuous function
if and only if $\tau_{1} \subseteq \tau_{2}$
\begin{definition}[\textbf{homeomorphism}]
A mapping $f:X \rightarrow Y$ between topological spaces, will be called an \textbf{homeomorphism} if:
\begin{itemize}
\item it is bijective (that is: ``$1-1$" and onto)
\item both $f$ and $f^{-1}$ are continuous
\end{itemize}
\end{definition}
\textbf{\underline{Exercises:}}    \\
$\mathbf{9.}$ Prove that a mapping $f: X \rightarrow Y$ is a homeomorphism if and only if the mapping $f^{-1}: Y \rightarrow X$ is defined and the mappings
$f, \ f^{-1}$ are both open and closed at the same time.     \\
$\mathbf{10.}$ Show that in the set of all topological spaces, homeomorphism defines an equivalence \footnote{\textbf{Def.:} A (binary) relation $\thicksim$ on a set S, is called an
\textbf{equivalence relation on S} if it satisfies the following three conditions: \\
%\begin{itemize}
$\bullet$ \emph{Reflexivity}: $a \thicksim a$, $\ \forall a \in S$  \\
$\bullet$ \emph{Symmetry}: If $a \thicksim b$ then $b \thicksim a$    \\
$\bullet$ \emph{Transitivity}: If $a \thicksim b$ and $b \thicksim c$ then $a \thicksim c$
%\end{itemize}
} relation.   \\
Let us now present a couple of examples of homeomorphisms:
\begin{example}[\textsf{Open interval is homeomorphic to $\mathbb{R}$}]  \exlabel{opintervrealine}
Let the function
\begin{equation}
\begin{array}{c}
f: (-1,1) \rightarrow \mathbb{R}  \\
x \mapsto f(x) = \frac{x}{1-x^{2}}
\end{array}
\end{equation}
$f$ is bijective, continuous, with continuous inverse function (which is the inverse ?). It is an homeomorphism between the interval $(-1,1)$ and the whole real line
%\begin{center}
%\includegraphics[scale=0.3]{figure5-homeomex1}
%\end{center}
\end{example}
\begin{example}[\textsf{Unit open ball $\mathbb{B}^{n}$ is homeomorphic to $\mathbb{R}^{n}$}]
Let $\mathbb{B}^{n} \equiv B(0,1)$ denote the open unit ball of $\mathbb{R}^{n}$ centered at $x_{0}=0$. Define a map $F: \mathbb{B}^{n} \rightarrow \mathbb{R}^{n}$ by
$$
\vec{y} = F(\vec{x}) = \frac{\vec{x}}{1-|\vec{x}|^{2}}
$$
Note that $|F(\vec{x})| = \frac{|\vec{x}|}{1-|\vec{x}|^{2}} \longrightarrow \infty$ as $|\vec{x}| \longrightarrow 1$. $F$ is bijective and we can construct the inverse of $F$ which is found (verify!) to be:
$$
\vec{x} = F^{-1}(\vec{y}) = \frac{2\vec{y}}{1 + \sqrt{1+4|\vec{y}|^{2}}}
$$
Since both $F$ and $F^{-1}$ are continuous, $F$ is an homeomorphism between the open unit ball $\mathbb{B}^{n}$ and the whole of $\mathbb{R}^{n}$
\end{example}
The previous examples show that \underline{properties such as ``\emph{size}" or ``\emph{boundedness}" are} \underline{not topological properties}.
\begin{example}[\textsf{Homeomorphism between a cube $\mathcal{C}$ and a sphere $\mathcal{S}^{2}$}]
Let $\mathcal{S}^{2} = \{ \vec{r} = (x,y,z) \in R^{3} \diagup |\vec{r}| = 1 \}$ be the unit sphere in $\mathbb{R}^{3}$ centered at the origin and set
$\mathcal{C} = \{ \vec{r} = (x,y,z) \in R^{3} \diagup max(|x|,|y|,|z|) = 1 \}$ which is the cubical surface of side $2$ centered at the origin.

Let $\phi$ be the map \cite{Lee} that projects each point of $\mathcal{C}$ radially inwards to the sphere $\mathcal{S}^{2}$.
$$
\begin{array}{c}
\phi : \mathcal{C} \rightarrow \mathcal{S}^{2}  \\
   \\
\phi(\vec{r}) = \phi(x,y,z) = \frac{(x,y,z)}{\sqrt{x^{2}+y^{2}+z^{2}}} = \frac{x}{|\vec{r}|}\vec{i} + \frac{y}{|\vec{r}|}\vec{j} + \frac{z}{|\vec{r}|}\vec{k}
\end{array}
$$
More precisely, for any point $\vec{r} \in \mathcal{C}$, $\phi(\vec{r})$ is the unit vector in the direction of $\vec{r}$. This function is continuous by the usual arguments of
elementary calculus.
\begin{center}
\includegraphics[angle=270, scale=0.4]{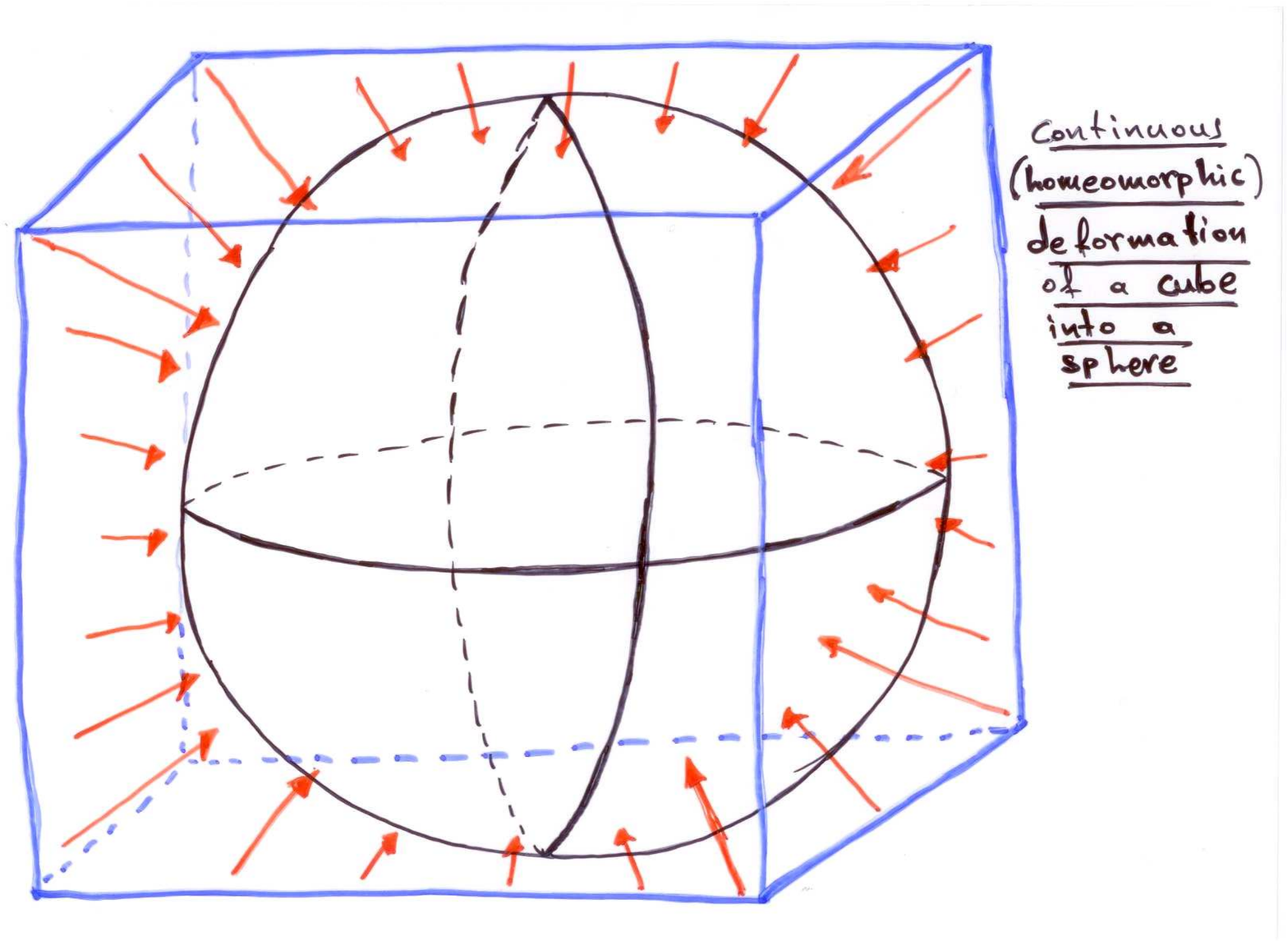}
\end{center}
Let us describe now the inverse function $\phi^{-1}$. This is given by
$$
\begin{array}{c}
\phi^{-1} : \mathcal{S}^{2} \rightarrow \mathcal{C} \\
   \\
\phi^{-1}(\vec{r}) = \phi^{-1}(x,y,z) = \frac{(x,y,z)}{max(|x|,|y|,|z|)} = \frac{x}{max(|x|,|y|,|z|)}\vec{i} + \frac{y}{max(|x|,|y|,|z|)}\vec{j} + \frac{z}{max(|x|,|y|,|z|)}\vec{k}
\end{array}
$$
and it describes a radial projection of the sphere $\mathcal{S}^{2}$ outwards to the cube $\mathcal{C}$.
\end{example}
The above example shows us that \underline{``\emph{shape}" is not a topological property}. But what (or: which) are finally the ``\emph{topological properties}" ?

The answer is simple and the idea may have already been motivated by the preceding examples: We define \textbf{topological properties} to be
\textbf{those preserved under homeomorphisms}. An homeomorphism defines a bijective correspondence not only between the points of the topological spaces $X, \ Y$ but also between the topologies themselves, i.e. between the collections of open and closed sets. Hence any property of the topological space $(X, \tau_{X})$ stated in terms of its topology (i.e. in terms of its collection $\tau_{X}$ of open sets) is also valid for any homeomorphic space $(Y, \tau_{Y})$, and is similarly stated in terms of its topology (i.e. in terms of the
$\tau_{Y}$ collection). Thus, homeomorphic spaces possess identical topological properties and are thus indistinguishable from this point of view. When we study
such topologically invariant properties we need not distinguish between homeomorphic spaces. the term invariant is used in the sense that topological properties are exactly those left invariant under homeomorphisms, which is equivalent to saying that \textbf{topological properties are those which are stated only in terms of the collection $\tau$ of open sets}. Note, in this connection, that the main task of topology was for a long time (and still is in some sense) to create an effective method of distinguishing between non-homeomorphic spaces. The definitions of topological properties such as compactness, connectedness, Hausdorff and more, have been the outcome of such attempts. Moreover -as we shall see in the sequel- such properties are direct abstractions of well known properties of the real numbers and more generally of metric spaces.

Finally let us note that the notion of continuity in topology resembles the notion of homomorphism of group theory while the notion of homeomorphism in topology finds its analogue in algebra in the notion of isomorphism.

Before continuing with some exercises lets see an example of non-homeomorphic spaces
\begin{example}
Here is a case of non-homeomorphic topological spaces
\begin{center}
\includegraphics[angle=270, scale=0.3]{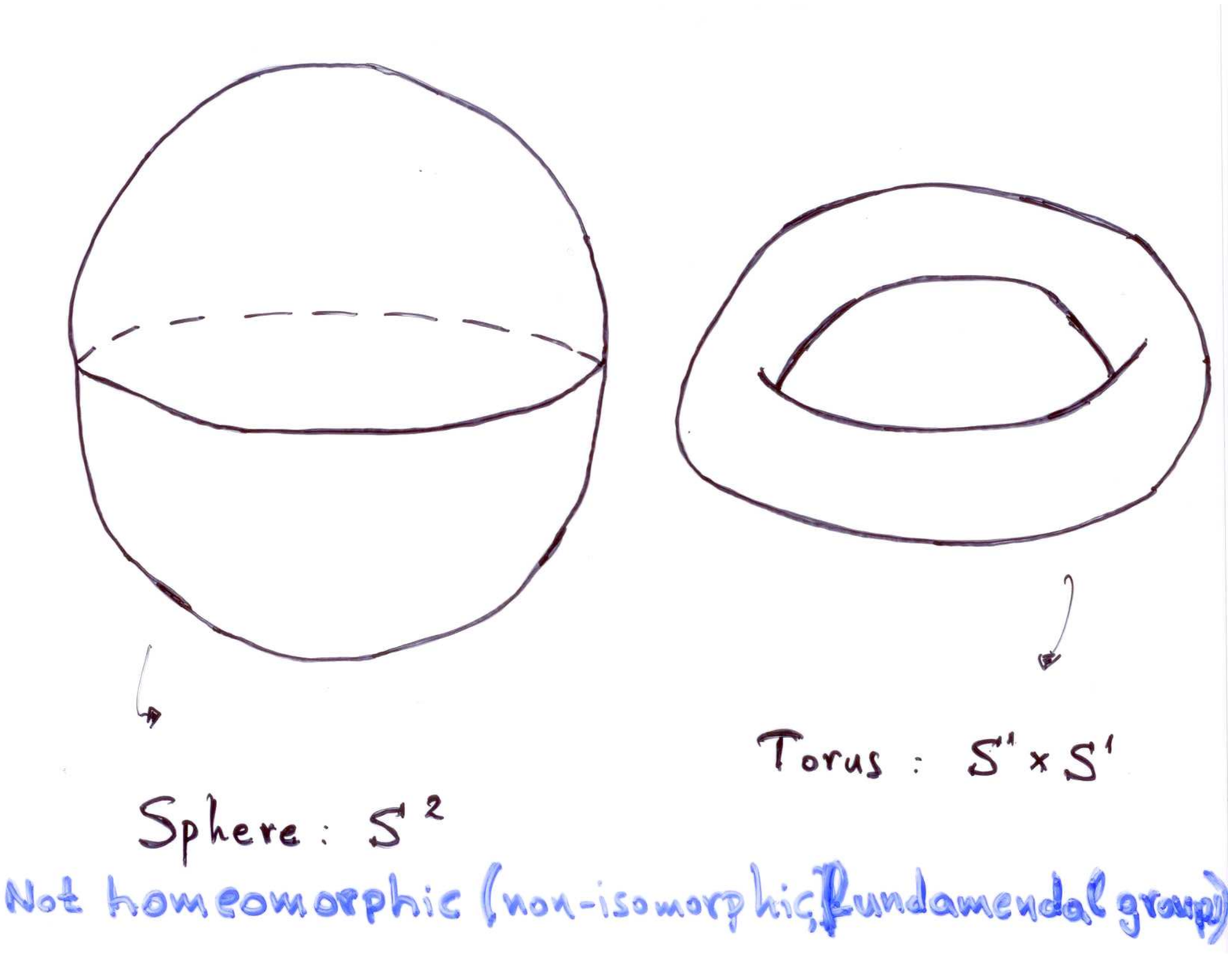}
\end{center}
Both the sphere $S^{2}$ and the torus $S^{1} \times S^{1}$ are considered as subspaces of $\mathbb{R}^{3}$ (and thus equipped with the relative topology)
\end{example}
\textbf{\underline{Exercises:}}  \\
$\mathbf{11.}$ Construct the inverse function $f^{-1}$ for the \exref{opintervrealine} and prove that it is continuous $\forall \ x \in \mathbb{R}$.   \\
$\mathbf{12.}$ Show that the function $f: (0,1) \rightarrow \mathbb{R}$ given by: $f(x) = \frac{2x-1}{x(x-1)}$ is an homeomorphism and determine its inverse $f^{-1}$. \\
$\mathbf{13.}$ Show that the function $g: [0,1) \rightarrow \mathcal{S}^{1}$ given by $g(t) = (\cos2 \pi t, \sin2 \pi t)$ is continuous and bijective but not an homeomorphism (\underline{Hint:} Show that its inverse function $g^{-1}: \mathcal{S}^{1} \rightarrow [0,1)$ is discontinuous at one point of the circumference)  \\
$\mathbf{14.}$ Consider the open disc $\mathbb{B}^{n-1} \subseteq \mathbb{R}^{n}$ centered at the origin, with radius $1$. Consider also the northern hemisphere (excluding the equator) $\mathcal{S}^{n-1}_{+}$ of the open sphere $\mathcal{S}^{n-1} \subseteq \mathbb{R}^{n}$ which is centered at the origin and has radius $1$. Construct an homeomorphism $\phi: \mathcal{S}^{n-1}_{+} \rightarrow \mathbb{B}^{n-1}$ and its inverse $\phi^{-1}: \mathbb{B}^{n-1} \rightarrow \mathcal{S}^{n-1}_{+}$. (\underline{Hint:} Take the ``vertical" projection of the hemisphere on the disc)   \\
$\mathbf{15.}$ Consider the unit circle $\mathcal{S}^{1} \subset \mathbb{R}^{2}$ and the perimeter of a square, whose length side equals $2$. Construct an homeomorphism between them (\underline{Hint:} Consider both figures centered at the origin of $\mathbb{R}^{2}$ and construct the outward and inward radial projections)   \\

Let us now turn to describe the notion of the induced topology:
\begin{proposition} \prlabel{indtopol}  $ $
\begin{itemize}
\item Let a function $f: (X, \tau_{X}) \rightarrow Y$ where $Y$ is a set (not necessarily equipped with any kind of topology). If we consider the collection
$$
\tau_{*} = \{ w \diagup w \subseteq Y : f^{-1}(w) \in \tau_{X} \}
$$
of subsets of $Y$, then we can show that $(Y, \tau_{*})$ becomes a topological space and furthermore that $f$ is a continuous function.
\item Let a function $f: X \rightarrow (Y, \tau_{Y})$ where $X$ is a set (not necessarily equipped with any kind of topology). If we consider the collection
$$
\tau_{*} = \{ f^{-1}(V) \diagup V \in  \tau_{Y} \}
$$
of subsets of $X$, then we can show that $(X, \tau_{*})$  becomes a topological space and that furthermore $f$ is a continuous function.
\end{itemize}
\end{proposition}
The topologies of the above proposition, determined by $\tau_{*}$, are frequently called in the literature as \textbf{induced topologies}. (\underline{Hint:} See
\cite{KoFo} p.5) for the definition of the notion of the preimage $f^{-1}(w)$ of a set).

Lets see an elementary application of the above proposition: Let $(X, \tau_{X})$ a topological space and consider any $A \subseteq X$ subset of $X$. Then the inclusion map
$i: A \hookrightarrow (X, \tau_{X})$ given by $i(x) = x$ for any $x \in A$, induces (according to the second statement of \prref{indtopol}) a topology on $A$. Such a topology
on $A$ is frequently called in the literature as \textbf{relative topology} on $A$. In this way, any subset of a topological space may be considered as a topological space
itself. \\
\textbf{\underline{Exercises:}}    \\
$\mathbf{16.}$ Prove both statements of \prref{indtopol}    \\
$\mathbf{17.}$ Let $A \subseteq X$ of $(X, \tau_{X})$. Consider the collection of subsets of $A$ of the form $\{ A \cap V \diagup V \in \tau_{X} \}$.
Show that the above collection forms a topology on $A$ and that this topology is \textbf{equivalent} (in the sense that it consists of exactly the same open sets) to the relative
topology on $A$.

\subsection{Products}
Suppose $(X_{i}, \tau_{X_{i}})$ for $i=1, 2, ..., n$ are topological spaces with $\{ U_{a_{i}} \}$ the corresponding bases for the topologies. Then, the set-theoretic Cartesian product
$X_{1} \times X_{2} \times ... \times X_{n}$ can be equipped with a topology generated by the base
$$
\mathcal{B}_{\times} = \{ U_{a_{1}} \times U_{a_{2}} \times ... \times U_{a_{n}} \}
$$
The topology constructed in this way will be called \textbf{product topology} and the set $X_{1} \times X_{2} \times ... \times X_{n}$ equipped with the product topology will be called \textbf{product space}.     \\
\textbf{\underline{Exercises:}}   \\
$\mathbf{18.}$ What do the open sets in the product space look like?    \\
$\mathbf{19.}$ Show that the product topology on $\mathbb{R}^{n} = \mathbb{R} \times \mathbb{R} \times ... \times \mathbb{R}$ (``\emph{box topology}") is the same (or: equivalent, in the sense that it consists of exactly the same open sets) as the metric topology (``\emph{ball topology}") induced by the Euclidean distance function of $\mathbb{R}^{n}$ i.e. the Euclidean
topology  \\
$\mathbf{20.}$ Consider the projection functions $p_{i}: \prod_{k=1}^{n}X_{k} \rightarrow X_{i}$ defined by $p_{i}((x_{1}, x_{2}, ..., x_{n})) = x_{i}$ for $i=1, ..., n$. Prove that $p_{i}$ are continuous functions for any $i=1, ..., n$   \\ \\

From the beginnings of topology it was clear that there was something quite special about the closed and bounded interval $[a, b]$ of the real line. Fundamental theorems of Calculus like the \emph{Intermediate Value Theorem} \footnote{If $f:[a,b] \rightarrow \mathbb{R}$ is continuous, then for any real number $r$ between $f(a)$ and $f(b)$, there is some $c \in [a, b]$ such that: $f(c) = r$} or like the \emph{Extreme Value Theorem} \footnote{A continuous real-valued function on $[a, b]$ attains a maximum and a minimum value} may be interpreted to depend rather on the topological properties of the closed and bounded interval $[a,b]$ of real numbers, than on the properties of continuity itself. The question about the correct abstractions, capable of leading to definitions of sufficient generality to be applicable to any topological space, was hovering for quite a while. The outcome of such investigations was that the topological  property on which the intermediate value theorem depends upon is the notion of \emph{Connectedness} while the topological property upon which the extreme value theorem is based is the notion of \emph{Compactness}.

\subsection{Compactness}
We start with the more abstract \footnote{or less ... ``natural" if you prefer} notion of Compactness. For a long time it was not clear how to abstract the properties of $[a,b]$ supporting the extreme value theorem. Although it was clear that closedness and boundedness were of central importance, they were never considered as ``good" properties since boundedness for ex. does not even make sense for a general topological space. For a while it was thought that the crucial property was that, every infinite subset of $[a,b]$ has a limit point, and it was this property which was initially dignified with the name of compactness. Later it was realized that this formulation does not lie at the heart of the matter if we look outside metric spaces.
The classic theorem of Heine-Borel \footnote{care must be given to the fact that the theorem does not hold in an arbitrary metric space} of real analysis asserts that a subset $D \subset \mathbb{R}^{n}$ is closed and bounded if and only if any open covering of $D$ has a finite subcovering. This theorem has extraordinarily profound consequences and its history goes back to the development of the fundamentals of real analysis. As with most ``good" theorems, it was its conclusion which was finally abstracted and became a definition.
Lets start with the important concept of covering:
\begin{definition}
Given a set $X$ and a subset $A \subseteq X$, a \textbf{cover(ing)} for $A$ is a family of subsets $\mathcal{U} = \{ U_{i} \diagup i \in I \}$ of $X$ such that
$A \subseteq \cup_{i \in I}U_{i}$. A subcover(ing) of a given cover $\mathcal{U}$ for $A$ is a subfamily $\mathcal{V} \subset \mathcal{U}$ which still forms a cover for $A$. Finally,
given two coverings $\mathcal{U} = \{ U_{\alpha} \}$ and $\mathcal{V} = \{ V_{\beta} \}$ of $A \subseteq X$, we will say that $\mathcal{V}$ \textbf{refines} $\mathcal{U}$ (or that
$\mathcal{V}$ is finer than $\mathcal{U}$) if each set $V_{\beta}$ is contained in some set $U_{\alpha}$, $\alpha = \alpha(\beta)$.
\end{definition}
We come now to the definition of Compactness for a general topological space:
\begin{definition}[\textbf{Compactness}]
A topological space $X, \tau$ is said to be \textbf{compact} if every open covering of $X$ has a finite subcovering.
\end{definition}
A subset $A \subset X$ will be called \textbf{compact subset} if $A$ will be a compact space (in the relative topology).
We underline here, that in metric spaces the notion of compactness can be formulated in terms which are customary from mathematical analysis:
\begin{theorem}[Compactness in metric spaces]
The metric space $X$ is compact \footnote{with topology induced by the metric} if and only if any sequence $\{ x_{n} \}$ has a convergent subsequence
\end{theorem}
The following prop., summarizes a few important properties of compactness:
\begin{proposition}
Let $(X, \tau_{X})$ a compact topological space.
\begin{itemize}
\item Any closed subset of a compact space is a compact subset, i.e. compactness is \textbf{closed-hereditary}
\item If $f:X \rightarrow Y$ a continuous function to the topological space $(Y, \tau_{Y})$ and $A \subset X$ is a compact subset of $X$ and then $f(A)$ is a compact subset of $Y$
\item (\textbf{Tihonov's theorem:}) Any product (finite or infinite) of compact spaces is compact, i.e. compactness is \textbf{productive}
\end{itemize}
\end{proposition}
\begin{example}
$[a,b]$ is compact for any $a,b \in \mathbb{R}$. This is the content of the Heine-Borel theorem.
\end{example}
\begin{example}
$(a,b)$ is not compact. To see this it is enough to consider the (infinite) open covering $\{ (a + \frac{1}{n}, b) \diagup n \in \mathbb{N}, n > \frac{1}{b-a} \}$
\end{example}
\textbf{\underline{Exercises:}}   \\
$\mathbf{21.}$ Show that a closed box and a closed ball in $\mathbf{R}^{n}$ are compact sets (\underline{Hint:} Instead of appealing to the Heine-Borel theorem, use the properties
of compactness under products and homeomorphisms)   \\
$\mathbf{22.}$ Show that $S^{1} \subset \mathbb{R}^{2}$ is compact (\underline{Hint:} Consider the mapping $[0,1] \rightarrow S^{1}$ given by $(x,y) = (cos 2 \pi t, sin 2 \pi t)$)   \\
$\mathbf{23.}$ Show that $\mathbb{R}$ (equipped with its usual ``Euclidean" topology) is not a compact space

\subsection{Connectedness}
Connectedness abstracts the idea of the ``wholeness" or the ``integrity" of a geometrical object and the abstraction is carried in a way to ensure that connected spaces will behave
similarly to the intervals of the real line (so for example a continuous function on a connected space will satisfy the intermediate value theorem). Lets describe the general idea.
Consider a topological space $(X, \tau_{X})$ and its subsets $A, B$
\begin{definition}
The sets $A,B$ are said to be \textbf{separated} if
$$
\overline{A} \cap B = A \cap \overline{B} = \emptyset
$$
\end{definition}
For example if $X=\mathbb{R}$ (with the Euclidean topology) the number line and $A=(a,b) \ $, $B=(b,c)$ are intervals with $a < b < c$ then
$A$ and $B$ are separated. But if $A=(a,b] \ $, $B=(b,c)$ then $A$ and $B$ are not separated.
\begin{definition}
A space $X$ is said to be \textbf{disconnected} if it can be represented as the union of two nonempty separated sets.
\end{definition}
A space not satisfying the condition of the above definition is said to be connected. Thus a connected space cannot be represented as the union of
two non-empty, separated sets. A subset $A \subset X$ will be called \textbf{connected subset} if $A$ will be a connected subspace (i.e. a connected space
in the relative topology).

The next lemma gives us an equivalent description of the notion of \textbf{connectedness} and is used quite often:
\begin{lemma}
A topological space $X$ is \textbf{connected} if and only if any one of the following equivalent conditions holds:
\begin{itemize}
\item It cannot be decomposed into the union of two non-empty, disjoint, open sets
\item It cannot be decomposed into the union of two non-empty, disjoint, closed sets
\item The only subsets of $X$ which are open and closed at the same time are the empty set $\emptyset$ and $X$ itself
\end{itemize}
\end{lemma}
Finally we record a couple of important properties related to connectedness:
\begin{proposition}
Connectedness is preserved by continuous maps and arbitrary products:
\begin{itemize}
\item If $X$ is a connected space and $f:X  \rightarrow Y$ a continuous map then $f(X)$ is connected in $Y$ (a connected subset of $Y$).
\item Any product (finite or infinite) of connected spaces must be connected i.e. connectedness is productive.
\end{itemize}
\end{proposition}
\begin{example}
All intervals on the real line $[a,b]$, $(a,b)$, $[a,b)$, $(a,b]$ etc are connected
\end{example}
\begin{example}
The set of rational numbers $Q = \{ \frac{p}{q} \diagup p,q: \ relatively \ prime \ integers, \ q \neq 0 \}$ on the real line is disconnected. To see this it suffices to consider
an arbitrary irrational real number $\alpha$ and the sets:
$$
\begin{array}{cccc}
U = \{ x: x \in \mathbb{Q},  x < \alpha\}    &    &   &  V = \{ x: x \in \mathbb{Q},  x > \alpha\}
\end{array}
$$
these are non-empty, open and disjoint. Moreover $\mathbb{Q} = U \bigcup V$
\end{example}
\textbf{\underline{Exercises:}}   \\
$\mathbf{24.}$ Prove that in a space endowed with the discrete topology, any set (except the one-point sets) is disconnected.   \\
$\mathbf{25.}$ Prove that $S^{1}$ is connected.  \\
$\mathbf{26.}$ Prove that the set of all rationals $\mathbb{Q} = \{ \frac{p}{q} \ / p, q: \ relatively \ prime \}$ on the real line is a disconnected set.   \\
$\mathbf{27.}$ Show that the only connected subspaces of the rationals $\mathbb{Q}$ are the one-point subsets.   \\
$\mathbf{28.}$ Show that balls and boxes (open or closed) in $\mathbb{R}^{n}$ are connected.

\subsection{Hausdorff}
The idea behind the definition of Hausdorff property is to have a kind of criterion regarding the abundance of our supply of open sets.
\begin{definition}[\textbf{Hausdorff}]
A topological space $(X, \tau)$ will be called a \textbf{Hausdorff space} or a $\mathbf{\mathcal{T}_{2}}$\textbf{-space} if for any two elements $x, y \in X$ we can always find disjoint neighborhoods for them. Formally:
$$
Given \ x,y \in X \ with \ x \neq y \ \ \exists \ neighborhhoods \ V_{x}, V_{y} \ such \ that: \ V_{x} \cap V_{y} = \emptyset
$$
\end{definition}
The next proposition records a couple of important properties of Hausdorff spaces:
\begin{proposition}
The Hausdorff (or $\mathcal{T}_{2}$) property is:
\begin{enumerate}
\item Hereditary, i.e. any subset of a Hausdorff space is a Hausdorff space (in the relative topology)
\item Productive, i.e. any product of Hausdorff spaces is a Hausdorff space (in the product topology)
\item In a Hausdorff space any pair of disjoint, compact sets can also bee separated by disjoint neighborhoods
\end{enumerate}
Furthermore, the next statement justifies the motivation for introducing the notion of the Hausdorff property
\begin{itemize}
\item[$4.$] In a Hausdorff space any compact subset is closed
\end{itemize}
\end{proposition}
The third statement of the above proposition is a example of the general (intuitive rule) that \emph{compact sets often} (under suitable assumptions) \emph{behave as points}.   \\
Let´s see now a couple of examples of the Hausdorff property:
\begin{example} \exlabel{metrtopolHaus}
Any metric space equipped with the metric topology, is a Hausdorff space
\end{example}
\begin{example}
Consider the set $X = \{ 1, 2, 3 \}$ with topology $\tau = \big\{\emptyset, \{ 1 \}, \{ 2, 3 \}, \{ 1, 2, 3 \} \big\}$. This is not a Hausdorff space
since the elements $2$ and $3$ do not have disjoint neighborhoods.
\end{example}
\textbf{\underline{Exercises:}}    \\
$\mathbf{29.}$ A topological space equipped with the trivial topology is not Hausdorff if it contains more than one point.   \\
$\mathbf{30.}$ Prove that the only Hausdorff topology on a finite set is the discrete topology   \\
$\mathbf{31.}$ Prove that $\forall x \in X$ ($X$: Hausdorff space) we have $\bigcap_{\alpha \in \mathcal{A}}V_{\alpha} = \{ x \}$ where $V_{\alpha}$ is an open set containing $x$
and $\mathcal{A}$ is the family of all open sets containing $x$ (Notice that $\mathcal{A}$ maybe an infinite family and that $\bigcap_{\alpha \in \mathcal{A}}V_{\alpha}$ does not
necessarily have to be open itself) \\
$\mathbf{32.}$ Show that in a Hausdorff space any one-point subset is closed.  \\
$\mathbf{33.}$ Is the continuous image of a Hausdorff space necessarily a Hausdorff space ? Prove or state a counterexample.   \\
$\mathbf{34.}$ Prove the assertion of \exref{metrtopolHaus}   \\
$\mathbf{35.}$ Show that the Hausdorff property is in fact a topological property i.e. that it is preserved under homeomorphisms.

\newpage

\section{Topological and Differentiable Manifolds}

We begin with the definition of the notion of topological manifold of dimension $n$
\begin{definition}[$n$-dim \textbf{Topological Manifold}] \delabel{topmanif}
An $\mathbf{n}$\textbf{-dimensional topological manifold} $\mathcal{N}$, is a Hausdorff, connected, topological space $(\mathcal{N}, \tau_{\mathcal{N}})$ which is locally homeomorphic to the Euclidean space $\mathbb{R}^{n}$. This means that any point $P \in \mathcal{N}$ is contained in some neighborhood $V_{P} \subseteq \mathcal{N}$, homeomorphic to a domain $U = \phi(V_{P}) \subseteq \mathbb{R}^{n}$ of the Euclidean space.
\end{definition}
\begin{center}
\includegraphics[angle=270, scale=0.5]{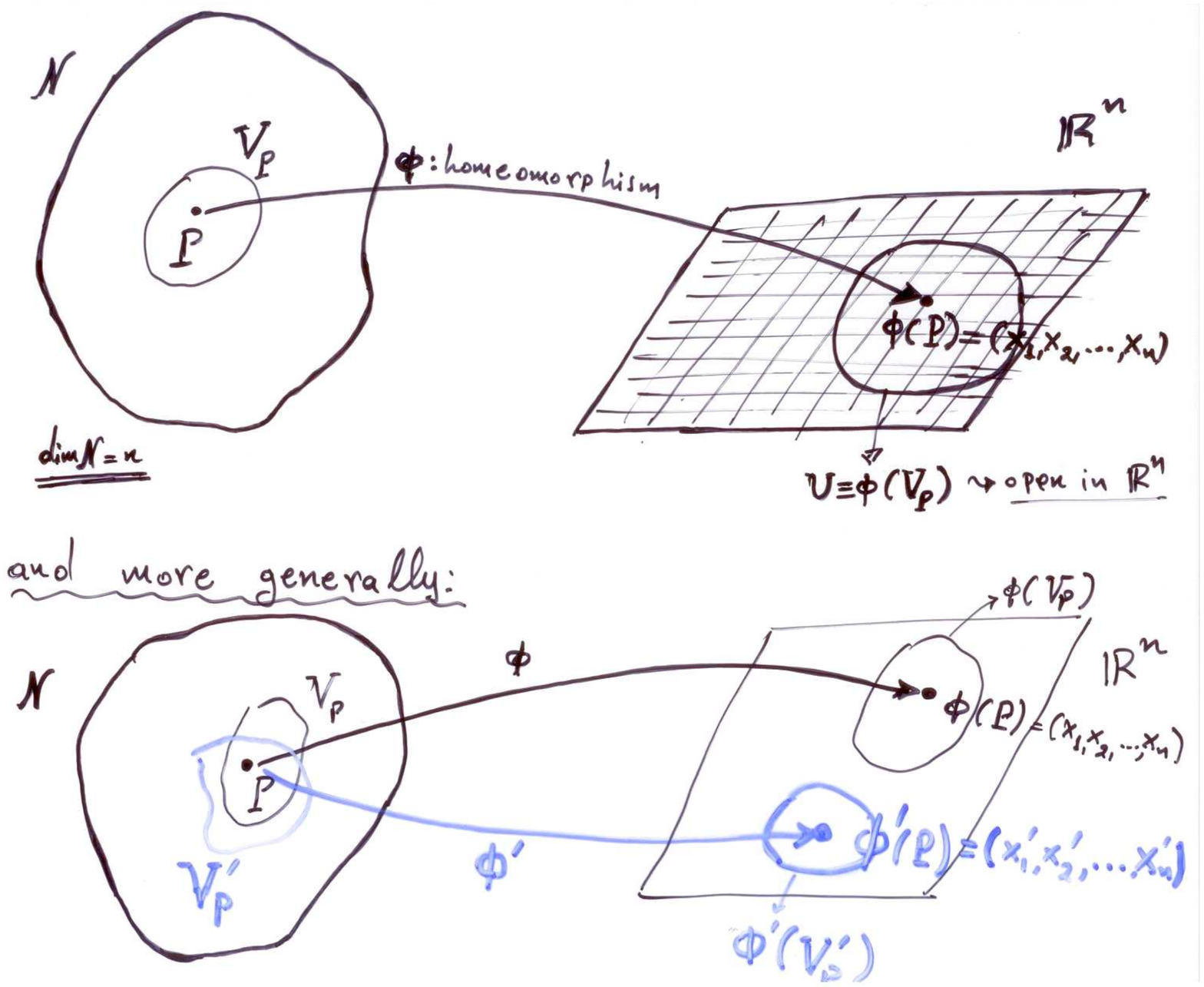}
\end{center}
Thus when $\mathcal{N}$ is an $n$-dim. topological manifold, we can find in $\mathcal{N}$ a system of open sets $V_{i}$ numbered by finitely (or infinitely) many indices $i$ and a system of homeomorphisms $\phi_{i}: V_{i} \rightarrow \phi(V_{i}) \equiv U_{i} \subseteq \mathbb{R}^{n}$ of the open sets $V_{i}$ on the open domains $U_{i}$. The system of the open sets $\{ V_{i} \}$ must cover the space $\mathcal{N}$ i.e. $\mathcal{N} = \bigcup_{i}V_{i}$ and the domains $U_{i}$ may, in general, intersect each other.
\begin{definition}[\textbf{local chart}]
If $\mathcal{N}$ is a topological manifold, any pair $(V, \phi)$ will be called a \textbf{local chart}, where $V$ is an open subset of $\mathcal{N}$ and $\phi: V \rightarrow U \equiv \phi(V) \subseteq \mathbb{R}^{n}$ an homeomorphism onto an open domain $U$ of $\mathbb{R}^{n}$
\end{definition}
A local chart $(V, \phi)$ associates to any point $P \in V$ the $n$-tuple $\mathbf{\phi(P) = (x_{1}, x_{2}, ..., x_{n})}$, called \textbf{coordinates of} $\mathbf{P}$ \textbf{with respect to the local chart} $\mathbf{(V, \phi)}$.
\begin{definition}[$\mathbf{\mathcal{C}^{r}}$\textbf{-Atlas}]  \delabel{atlas}
A collection of charts $\mathbf{\mathcal{A} = \{ (U_{i}, \phi_{i}) \}_{i \in I}}$ will be called a $\mathbf{\mathcal{C}^{r}}$\textbf{-atlas for the topological manifold} $\mathbf{\mathcal{N}}$ if the following conditions are valid
\begin{itemize}
\item $\mathcal{N} = \bigcup_{i \in I} U_{i}$ i.e. the open sets $\{ U_{i} \}_{i \in I}$ form a covering of $\mathcal{N}$
\item For any intersection $U_{i} \bigcap U_{j}$ the \textbf{transition functions}
\begin{equation}   \eqlabel{loccoordchang1}
\begin{array}{c}
\phi_{ij} \equiv \phi_{i} \circ \phi_{j}^{-1} : \phi_{j}(U_{i} \bigcap U_{j}) \stackrel{\phi_{i} \circ \phi_{j}^{-1}}{\longrightarrow} \phi_{i}(U_{i} \bigcap U_{j}) \subset \mathbb{R}^{n} \\
%     \\
\phi_{ji} \equiv \phi_{j} \circ \phi_{i}^{-1} : \phi_{i}(U_{i} \bigcap U_{j}) \stackrel{\phi_{j} \circ \phi_{i}^{-1}}{\longrightarrow} \phi_{j}(U_{i} \bigcap U_{j}) \subset \mathbb{R}^{n}
\end{array}
\end{equation}
are $\mathcal{C}^{r}$-differentiable $ \ \forall \ i,j \in I$. Equivalently we say that the charts $(U_{i}, \phi_{i})$, $(U_{j}, \phi_{j})$ are $\mathbf{\mathcal{C}}^{r}$\textbf{-compatible} $\forall \ i,j \in I$. Visually we have
\end{itemize}
\end{definition}
\begin{center}
\includegraphics[angle=270, scale=0.5]{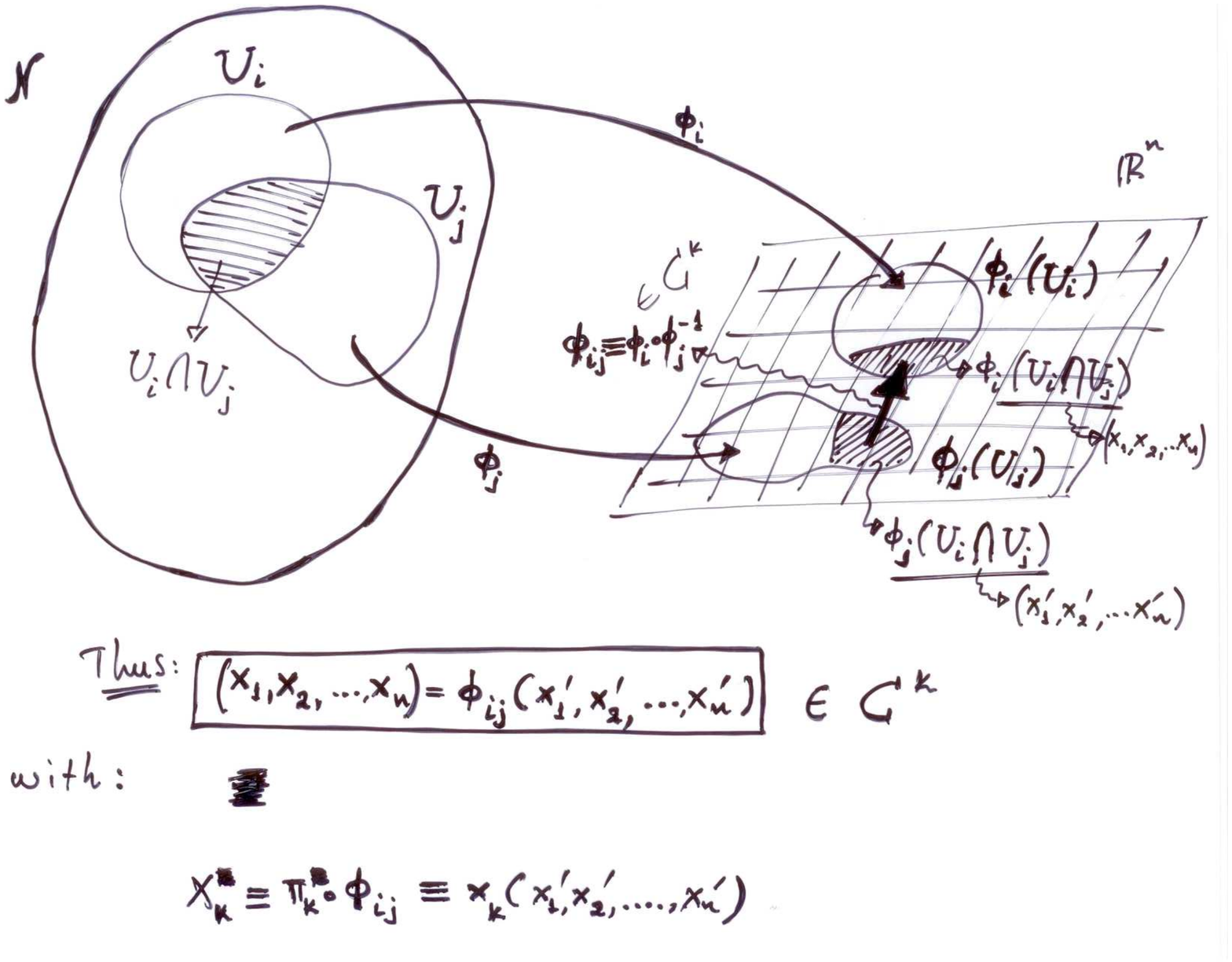}
\end{center}
In other words we have the $\mathcal{C}^{r}$-differentiable transition functions \equref{loccoordchang1}:
%\begin{equation}  \eqlabel{loccoordchang2}
$$ (x_{1}, x_{2}, ..., x_{n}) = \phi_{ij}(x_{1}^{'}, x_{2}^{'}, ..., x_{n}^{'}) \Leftrightarrow (x_{1}^{'}, x_{2}^{'}, ..., x_{n}^{'}) = \phi_{ji}(x_{1}, x_{2}, ..., x_{n}) $$
%\end{equation}
which can also be written as $x_{k} \equiv x_{k}(x_{1}^{'}, x_{2}^{'}, ..., x_{n}^{'}) = \pi_{k} \circ \phi_{ij}(x_{1}^{'}, x_{2}^{'}, ..., x_{n}^{'})$

We have two notable remarks regarding the above definition:  \\
$\mathbf{\centerdot}$ At this point we should underline the fact that \emph{in flat spaces} \footnote{here the term implies: vector spaces} relations like \equref{loccoordchang1} describe the change of coordinates between different bases of the space and \emph{are linear relations} whereas \emph{in the general case of a manifold these are non-linear relations}.   \\
$\mathbf{\centerdot}$ The invertible functions $\phi_{ij} = \phi_{ji}^{-1}$ are homeomorphisms (since by def. they are compositions of homeom.) between open domains of $\mathbb{R}^{n}$.
Moreover we have $\forall P \in \phi_{i}(U_{i} \bigcap U_{j})$ and $\forall P \in \phi_{j}(U_{i} \bigcap U_{j})$
$$
\begin{array}{cccccccc}
\frac{\partial (x_{1}, x_{2}, ..., x_{n})}{\partial (x_{1}^{'}, x_{2}^{'}, ..., x_{n}^{'})} \neq 0  & & & & & & &  \frac{\partial (x_{1}^{'}, x_{2}^{'}, ..., x_{n}^{'})}{\partial (x_{1}, x_{2}, ..., x_{n})} \neq 0
\end{array}
$$
i.e. the Jacobians of the transition functions are non-vanishing.
\begin{definition}[\textbf{Equivalent atlases}]   \delabel{equivalencyatlases}
Let $\mathcal{N}$ be a topological manifold and let $\mathcal{W} = \{ \mathcal{A} \diagup \mathcal{A} : atlas \ on \ \mathcal{N} \}$ be the collection of all $\mathcal{C}^{r}$-atlases defined on $\mathcal{N}$. Two $\mathcal{C}^{r}$-atlases $\{ (U_{\alpha}, \phi_{\alpha}) \}$ and $\{ (V_{\beta}, \psi_{\beta}) \}$ of $\mathcal{W}$ will be called \textbf{equivalent} if any two of their charts $(U_{\alpha}, \phi_{\alpha})$ and $(V_{\beta}, \psi_{\beta})$ are $\mathbf{\mathcal{C}}^{r}$-compatible for any value of $\alpha, \beta$. In other words: Two $\mathcal{C}^{r}$-atlases will be equivalent iff their union is a $\mathcal{C}^{r}$-atlas.
\end{definition}
The relation $\mathcal{R}$ introduced in the above definition is an equivalence relation  and under this relation the set $\mathcal{W}$ of all $\mathcal{C}^{r}$-atlases defined on $\mathcal{N}$ decomposes into equivalence classes:
$$
\overline{\mathcal{W}} \equiv \mathcal{W}/\mathcal{R} = \{ [\mathcal{A}] \diagup [\mathcal{A}]: equivalence \ classes \ of \ \mathcal{R} \}
$$
\begin{definition}[$\mathbf{\mathcal{C}^{r}}$\textbf{-structure}]
Each equivalence class $[\mathcal{A}]$ of $\overline{\mathcal{W}} \equiv \mathcal{W}/\mathcal{R}$ will be called a $\mathbf{\mathcal{C}^{r}}$\textbf{-structure} on $\mathcal{N}$
\end{definition}
Each equivalence class of $\mathcal{C}^{r}$-atlases on $\mathcal{N}$ is determined by any of its representatives i.e. a given $\mathcal{C}^{r}$-structure on $\mathcal{N}$ can be restored by any of its $\mathcal{C}^{r}$-atlases. In other words, a $\mathcal{C}^{r}$-structure on $\mathcal{N}$ can be determined by specifying on $\mathcal{N}$ one $\mathcal{C}^{r}$-atlas from the $\mathcal{C}^{r}$-structure.

The union of all $\mathcal{C}^{r}$-atlases from a given $\mathcal{C}^{r}$-structure is also a $\mathcal{C}^{r}$-atlas called \textbf{maximal}. Specifying a $\mathcal{C}^{r}$-structure is equivalent to specifying the maximal atlas. Sometimes the maximal atlas itself may be called a $\mathcal{C}^{r}$-structure.
\begin{definition}[$n$-dim, $\mathbf{\mathcal{C}^{r}}$\textbf{-Differentiable manifold}]   \delabel{diffmanif}
A Hausdorff, connected, topological space $\mathcal{N}$ equipped with a $\mathcal{C}^{r}$-structure will be called a $\mathbf{\mathcal{C}^{r}}$\textbf{-Differentiable manifold} (or: a differentiable manifold of class $\mathcal{C}^{r}$). The dimension $n$, of the target space $\mathbb{R}^{n}$ of the homeomorphisms $\phi$ of the charts, will be called the \textbf{dimension} of the $\mathcal{C}^{r}$-manifold i.e. $dim\mathcal{N} = n$.
\end{definition}
For $r=0$ we have the $\mathcal{C}^{0}$-manifolds which are nothing but the topological manifolds of \deref{topmanif}. For $r = 1, 2, ..., \infty$ we have the $\mathcal{C}^{r}$-manifolds sometimes abbreviated as \textbf{smooth manifolds}. If in the second condition of \deref{atlas} the homeomorphisms $\phi_{ij}, \phi_{ji}$ are real analytic mappings ($\mathcal{C}^{\omega}$-mappings) then we are speaking about $\mathcal{C}^{\omega}$-atlases, $\mathcal{C}^{\omega}$-structures and $\mathbf{\mathcal{C}^{\omega}}$\textbf{-manifolds} or \textbf{real analytic structures} and \textbf{real analytic manifolds} respectively.

%\begin{scriptsize}
\begin{remark}
$\bullet$ \underline{The dimension of a $\mathcal{C}^{0}$-manifold is an \textbf{invariant} of the manifold i.e.} \underline{it is independent of the choice of an atlas}. To see this it is enough to consider the following situation: Suppose that $\mathcal{N}$ admits the two different atlases (belonging to two possibly different $\mathcal{C}^{r}$-structures)
$$
\begin{array}{ccccccc}
\{ (U_{\alpha}, \phi_{\alpha}) \}, \ with \ \phi_{\alpha}: U_{\alpha} \rightarrow \mathbb{R}^{n}  &  &  & and  & &  &
\{ (V_{\beta}, \psi_{\beta}) \}, \ with \ \psi_{\beta}: U_{\beta} \rightarrow \mathbb{R}^{m}
\end{array}
$$
with $m \neq n$. Then there would be sets $U_{\alpha}, V_{\beta}$ (for some values of the indices) such that $U_{\alpha} \bigcap V_{\beta} \neq \emptyset$ and the mapping
$$
\varphi \circ \psi^{-1}: \psi(U_{\alpha} \bigcap V_{\beta}) \rightarrow \phi(U_{\alpha} \bigcap V_{\beta})
$$
would be an homeomorphism. But this would contradict Brouwer's theorem of topology, which states that the \textbf{non-empty, open sets $U \subseteq \mathbb{R}^{n}$ and $V \subseteq \mathbb{R}^{m}$ may be homeomorphic only in the case $m =n$}.

$\bullet$ Actually we can show more than that: A fundamental theorem of general topology asserts that: \textbf{if $\mathcal{M}$ and $\mathcal{N}$ are two homeomorphic manifolds (of class $\mathcal{C}^{r}$ with $r = 0, 1, 2, ..., \infty$) then $dim\mathcal{N} = dim\mathcal{M}$}. Consequently, \underline{the dimension of a di-} \underline{fferentiable (or even topological) manifold is a topological property of the manifold} itself.
\end{remark}
%\end{scriptsize}
%\begin{scriptsize}
\begin{remark}
Note that \underline{a $\mathcal{C}^{0}$-structure on any space $\mathcal{N}$} (i.e. the structure of a topological manifold on $\mathcal{N}$) \underline{is unique} (this follows directly from \deref{topmanif}, \deref{diffmanif}). On the other hand if $r \neq 0$ then \underline{$\mathcal{N}$ may admit several different $\mathcal{C}^{r}$-structures}. For example it is immediate from the definitions that a $\mathcal{C}^{r}$-differentiable manifold is also a $\mathcal{C}^{p}$-differentiable manifold for all values $0 \leq p \leq r$

Moreover, it has been proved that \textbf{if there exists on $\mathcal{N}$ at least one $\mathcal{C}^{r}$-structure ($r \geq 1$), then there exist infinitely many $\mathcal{C}^{r}$-structures on $\mathcal{N}$}.
\end{remark}
%\end{scriptsize}
%\begin{scriptsize}
\begin{remark}
Another worth noticing consequense of the definitions is that \underline{if an atlas} \underline{on a manifold $\mathcal{N}$ consists of only one chart} (i.e. if $\mathcal{N}$ is homeomorphic to a Euclidean domain), \underline{then $\mathcal{N}$ is both a smooth and an analytic manifold}, or equivalently: a $C^{r}$-differentiable manifold for any value of $r = 0, 1, 2, ..., \infty, \omega$.
\end{remark}
%\end{scriptsize}
\textbf{\underline{Exercises:}}    \\
$\mathbf{1.}$ Show that the binary relation $\mathcal{R}$ on the set $\mathcal{W}$ introduced in \deref{equivalencyatlases} is actually an equivalence relation.    \\
$\mathbf{2.}$ Consider the topological manifold $\mathcal{N} = \mathbb{R}$ equipped with the $\mathcal{C}^{\infty}$-structure determined by the $\mathcal{C}^{\infty}$-atlas consisting of the single chart $(\mathbb{R}, Id)$. Next consider $\mathbb{R}$ equipped with the $\mathcal{C}^{\infty}$-structure determined by the $\mathcal{C}^{\infty}$-atlas consisting of the single chart $(\mathbb{R}, \phi)$ where $\phi(x) = x^{3}$. Show that both of the above charts determine $\mathcal{C}^{\infty}$-atlases, that these atlases are non-equivalent,and therefore, the $\mathcal{C}^{\infty}$-structures determined by them are different.
\begin{example}[\textsf{$\mathbf{\mathbb{R}^{n}}$ as a differentiable manifold}] \ \\
$\mathbb{R}^{n}$ equipped with its usual (Euclidean) topology is an $n$-dimensional, $\mathcal{C}^{r}$-differentiable manifold for all values $r = 0, 1, 2, ..., \infty, \omega$.
(Thus it is topological, smooth and real analytic manifold). We may choose as a representative of its $\mathcal{C}^{r}$-structure the $\mathcal{C}^{r}$-atlas consisting of the single chart $(\mathbb{R}, Id)$ where $Id$ is the identity function on $\mathbb{R}$.
\end{example}
\begin{example}[\textsf{The graph of a real valued, continuous, vector function $f:\mathbb{R}^{n} \rightarrow \mathbb{R}$ as a differentable manifold}] \ \\
Let $f:\mathbb{R}^{n} \rightarrow \mathbb{R}$ be a continuous function and let
$$
\Gamma_{f} = \{ (x_{1}, x_{2}, ..., x_{n}, x_{n+1}) \diagup x_{n+1} = f(x_{1}, x_{2}, ..., x_{n}) \} \subseteq \mathbb{R}^{n+1}
$$
denote its graph. $\Gamma_{f}$ is an $n$-dimensional, $\mathcal{C}^{r}$-differentiable manifold for all values $r = 0, 1, 2, ..., \infty, \omega$. Its $\mathcal{C}^{r}$-structure is determined if we choose as representative the $\mathcal{C}^{r}$-atlas consisting of the single chart $(\Gamma_{f}, \phi)$ where the homeomorphism $\phi: \Gamma_{f} \rightarrow \mathbb{R}^{n}$ is given by $\phi(x_{1}, x_{2}, ..., x_{n}, x_{n+1}) = (x_{1}, x_{2}, ..., x_{n}) \in \mathbb{R}^{n}$ and its inverse homeomorphism $\phi^{-1}: \mathbb{R}^{n} \rightarrow \Gamma_{f}$ is given by $\phi^{-1}(x_{1}, x_{2}, ..., x_{n}) = (x_{1}, x_{2}, ..., x_{n}, f(x_{1}, x_{2}, ..., x_{n}))$. These are apparently both continuous, bijective functions
\end{example}
\begin{example}[\textsf{$\mathbf{GL(n,\mathbb{R})}$ and other \textbf{matrix groups} as differentiable manifolds}]  \ \\
$\bullet$ Let $\mathbf{M_{n,k}(\mathbb{R})}$ denote the set of all real $n \times k$ matrices. $M_{n,k}(\mathbb{R})$ may be identified to $\mathbb{R}^{n \times k}$ through the homeomorphism $\phi: M_{n,k}(\mathbb{R}) \rightarrow \mathbb{R}^{n \times k}$
\begin{equation}   \eqlabel{matricestorealspaces}
\left(
  \begin{array}{cccc}
    a_{11} & a_{12} & \ldots & a_{1k} \\
    a_{21} & a_{22} & \ldots & a_{2k} \\
    \vdots & \vdots & \ddots & \vdots \\
    a_{n1} & a_{n2} & \ldots & a_{nk}
  \end{array}
\right)  \mapsto (a_{11}, a_{12}, \ldots , a_{1k}, a_{21}, a_{22}, \ldots , a_{2k}, \ldots \ldots, a_{n1}, a_{n2}, \ldots , a_{nk})
\end{equation}
(The above function is an homeomorphism with respect to the usual (Euclidean) topology of $\mathbb{R}^{n \times k}$ and the induced topology of $M_{n,k}(\mathbb{R})$ through the above mapping). In this way $M_{n,k}(\mathbb{R})$ becomes an $(n \times k)$-dimensional, $\mathcal{C}^{r}$-differentiable manifold for all values $r = 0, 1, 2, ..., \infty, \omega$.

$\bullet$ Now we consider the \emph{general linear group} $\mathbf{GL_{n}(\mathbb{R})}$ which as a set consists of all invertible, $n \times n$, real matrices. Equivalently, the elements of $GL_{n}(\mathbb{R})$ are all $n^{2}$ real matrices whose determinant is different than $0$. This is an open (verify!) subset of $M_{n,n}(\mathbb{R})$. It fails to be a differentiable manifold because it is a disconnected topological space. Let us see why: We consider the mapping $det: M_{n,n}(\mathbb{R}) \rightarrow \mathbb{R}$. This is a continuous mapping (since it is a polynomial function). The set $(-\infty, 0) \bigcup (0, \infty)$ is a disconnected subset of $\mathbb{R}$. For the inverse image of this set under the $det$ mapping, we have
\begin{small}
$$
\emptyset = det^{-1} \big( (-\infty, 0) \bigcap (0, \infty) \big) = det^{-1} \big( (-\infty, 0) \big) \bigcap det^{-1} \big( (0, \infty) \big) = GL^{-}_{n}(\mathbb{R}) \bigcap GL^{+}_{n}(\mathbb{R})
$$
\end{small}
$GL^{+}_{n}(\mathbb{R})$ and $GL^{-}_{n}(\mathbb{R})$ are the subsets of $GL_{n}(\mathbb{R})$ which consist of matrices with positive and negative determinants respectively. These are non-empty, open (since they are both defined as inverse images of the open sets $(0, \infty)$, $(-\infty, 0)$ respectively), disjoint and apparently $GL_{n}(\mathbb{R}) = GL^{+}_{n}(\mathbb{R}) \bigcup GL^{-}_{n}(\mathbb{R}$. Consequently, $GL_{n}(\mathbb{R})$ is disconnected.

$\bullet$ Its two connected (why?) components $\mathbf{GL^{+}_{n}(\mathbb{R})}$ and $\mathbf{GL^{-}_{n}(\mathbb{R})}$ are each $n^{2}$-dimensional, $\mathcal{C}^{r}$-differentiable manifolds for all values $r = 0, 1, 2, ..., \infty, \omega$. For each one of them, the $\mathcal{C}^{r}$-structure is determined by the $\mathcal{C}^{r}$-atlas consisting of the single chart $(GL^{\pm}_{n}(\mathbb{R}), \phi_{\pm})$ where $\phi_{\pm}$ are the suitable restrictions respectively of the $\phi$ homeomorphism given in \equref{matricestorealspaces}.
\end{example}
The previous examples were more or less trivial, in the sense that the differential structures were all determined by atlases consisting of a single chart each. But the real interesting cases -and in fact the real motivation behind the introduction of the notion of differentiable manifold- are the ones in which it is impossible to cover the whole space with a single chart. In such a case we have to pick an open covering of the space, determine the homeomorphisms from the open sets of the space to suitable real domains and finally to check the $\mathcal{C}^{r}$-compatibility between the different charts. The most typical example to start with is no other than the familiar $2$-Sphere $\mathcal{S}^{2}$.
\begin{example}[\textsf{The sphere $\mathbf{\mathcal{S}^{2}}$ as a differentiable manifold}]   \   \\
Lets consider the unit sphere, embedded in $\mathbb{R}^{3}$ (thus equipped with the relative topology as a subset of $\mathbb{R}^{3}$) and centered at the origin (for simplicity in computations)
$$
\mathcal{S}^{2} = \{ (x^{1}, x^{2}, x^{3}) \in \mathbb{R}^{3} \diagup (x^{1})^{2} + (x^{2})^{2} + (x^{3})^{2} = 1 \} \subset \mathbb{R}^{3}
$$

$\mathbf{(I).}$ If I consider the upper hemisphere $\mathcal{S}^{2}_{ \{x^{3} > 0\} }$, and the interior of the planar disc $D_{1,2} = \{ (x^{1}, x^{2}) \in \mathbb{R}^{2} \diagup (x^{1})^{2} + (x^{2})^{2} < 1 \}$, then the homeomorphism (why?)
\begin{equation}  \eqlabel{sphere1}
\begin{array}{c}
\phi_{1} : \mathcal{S}^{2}_{ \{x^{3} > 0\} } \rightarrow D_{1,2} \subset \mathbb{R}^{2}    \\
(x^{1}, x^{2}, x^{3}) \mapsto (x^{1}, x^{2})
\end{array}
\end{equation}
is the vertical, projection of the upper hemisphere $\mathcal{S}^{2}_{ \{x^{3} > 0\} }$ onto the open planar disc $D_{1,2}$. The situation is visually described in the following figure
\begin{center}
\includegraphics[angle=270, scale=0.3]{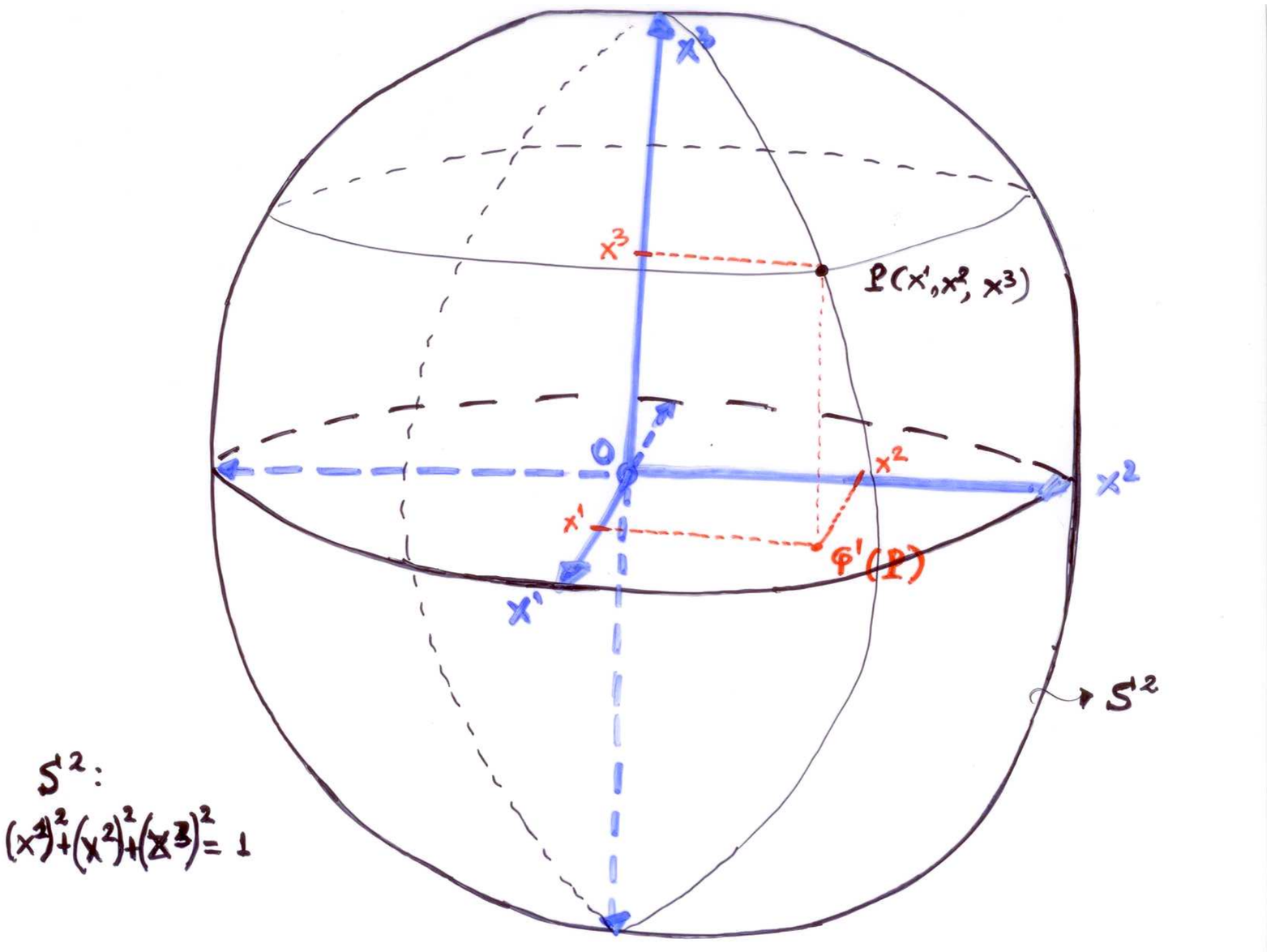}
\end{center}
Thus, $(\mathcal{S}^{2}_{ \{x^{3} > 0\} }, \phi_{1})$ is a local chart and it is clear that $\mathcal{S}^{2}$ can be covered with the domains of six such charts. \underline{We are going to prove that these six local charts constitute} \underline{a $\mathcal{C}^{\infty}$-atlas on $\mathcal{S}^{2}$}.

$\mathbf{(II).}$ Next, lets consider the stereographic projection of the sphere, from the north pole onto the plane $\pi (x^{3} = 0)$ coming through the equator
\begin{equation}    \eqlabel{sphere21}
\begin{array}{ccc}
\begin{array}{c}
\sigma_{1} : \mathcal{S}^{2} - \{ 0, 0, 1 \} \rightarrow \pi (x^{3} = 0) \subseteq \mathbb{R}^{2}   \\
(x^{1}, x^{2}, x^{3}) \mapsto \sigma_{1}(x^{1}, x^{2}, x^{3}) = (\zeta, \eta) \equiv (\frac{x^{1}}{1 - x^{3}}, \frac{x^{2}}{1 - x^{3}})
\end{array}    &  \rightsquigarrow   &    \begin{array}{c}
                                          Stereographic \ Projection  \\
                                          from \ the \ North \ Pole
                                          \end{array}
\end{array}
\end{equation}
The local chart $(\mathcal{S}^{2} - \{ 0,0,1 \}, \sigma_{1})$ given in \equref{sphere21} is visually described in the following
\begin{center}
\includegraphics[angle=270, scale=0.35]{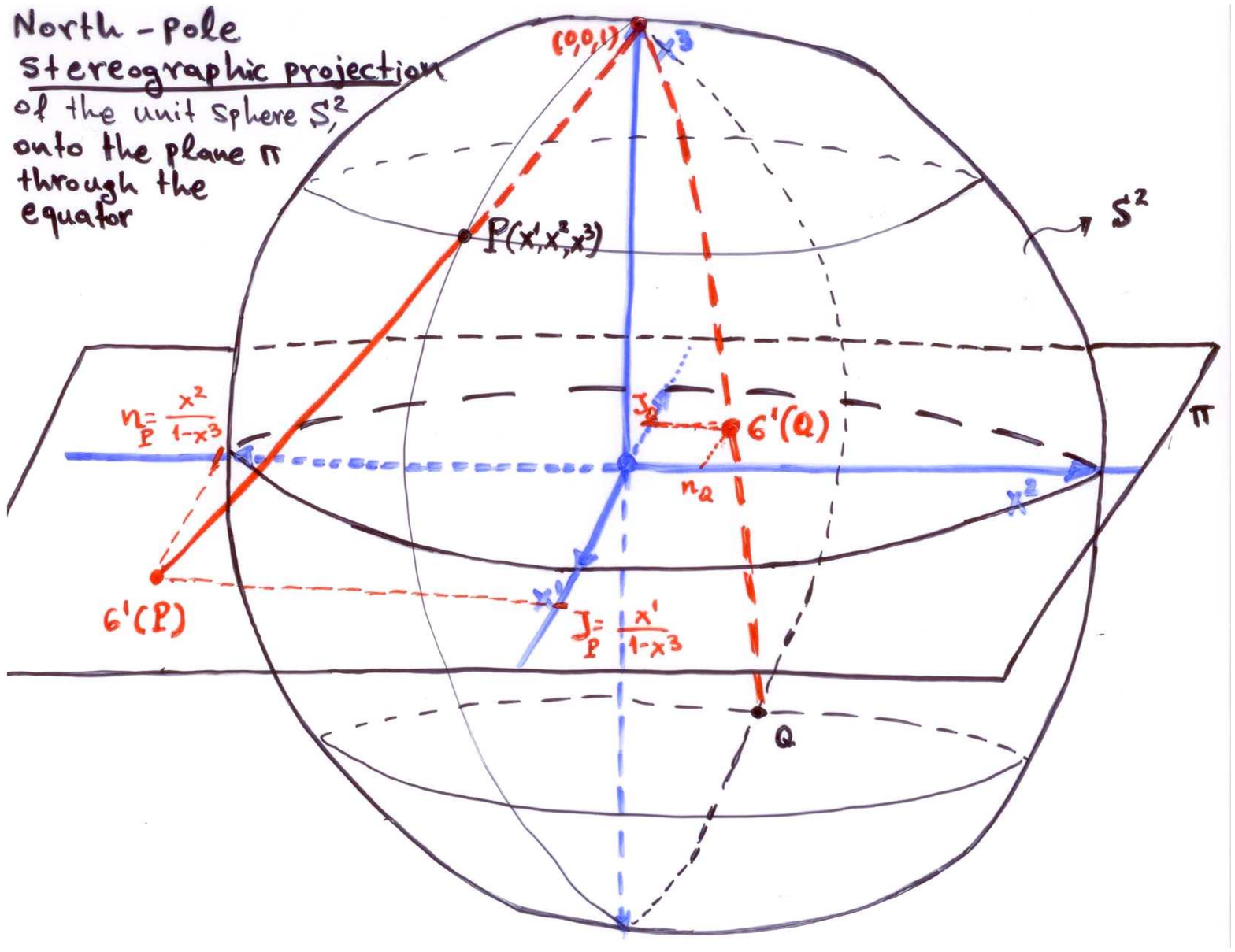}
\end{center}
\underline{We are going to prove that the chart given by \equref{sphere21}, supplemented with the chart de-} \underline{scribing the stereographic projection from the south pole}
\begin{equation}    \eqlabel{sphere22}
\begin{array}{ccc}
\begin{array}{c}
\sigma_{2} : \mathcal{S}^{2} - \{ 0, 0, -1 \} \rightarrow \pi (x^{3} = 0) \subseteq \mathbb{R}^{2}   \\
(x^{1}, x^{2}, x^{3}) \mapsto \sigma_{2}(x^{1}, x^{2}, x^{3}) = (\zeta', \eta') \equiv (\frac{x^{1}}{1 + x^{3}}, \frac{x^{2}}{1 + x^{3}})
\end{array}    &  \rightsquigarrow   &    \begin{array}{c}
                                          Stereographic \ Projection  \\
                                          from \ the \ South \ Pole
                                          \end{array}
\end{array}
\end{equation}
\underline{constitute a $\mathcal{C}^{\infty}$-atlas on $\mathcal{S}^{2}$}.

$\mathbf{(III).}$ Finally \underline{we are going to prove that the $\mathcal{C}^{\infty}$-atlas specified by \equref{sphere21}, \equref{sphere22} and} \underline{the $\mathcal{C}^{\infty}$-atlas specified by the $\phi_{i}$'s ($i = 1,2,...6$) of the form of the \equref{sphere1} are equivalent} or in other words that \underline{both atlases specify the same $\mathcal{C}^{\infty}$-structure}.   \\

\underline{\textbf{Proof of (I):}}    \\
It is enough to consider any two (out of the totally $6$) local charts with overlapping domains and show that they are $\mathcal{C}^{r}$-compatible.  Lets take $(\mathcal{S}^{2}_{ \{x^{3} > 0\} }, \phi_{1})$ and $(\mathcal{S}^{2}_{ \{x^{2} > 0\} }, \phi_{3})$ which are given by
\begin{equation}
\begin{array}{ccccc}
\begin{array}{c}
\phi_{1} : \mathcal{S}^{2}_{ \{x^{3} > 0\} } \rightarrow D_{1,2} \subset \mathbb{R}^{2}   \\
(x^{1}, x^{2}, x^{3}) \mapsto (x^{1}, x^{2})
\end{array}    &  &  \begin{array}{c}
                     with  \\
                     inverse:
                     \end{array}  &  &   \begin{array}{c}
                                         \phi_{1}^{-1} : D_{1,2} \rightarrow  \mathcal{S}^{2}_{ \{x^{3} > 0\} }  \\
                                         (x^{1}, x^{2}) \mapsto \big( x^{1}, x^{2}, \sqrt{1 - (x^{1})^{2} - (x^{2})^{2}} \big)
                                         \end{array}
\end{array}
\end{equation}
and
\begin{equation}
\begin{array}{ccccc}
\begin{array}{c}
\phi_{3} : \mathcal{S}^{2}_{ \{x^{2} > 0\} } \rightarrow D_{1,3} \subset \mathbb{R}^{2}   \\
(x^{1}, x^{2}, x^{3}) \mapsto (x^{1}, x^{3})
\end{array}    &  &  \begin{array}{c}
                     with  \\
                     inverse:
                     \end{array}  &  &   \begin{array}{c}
                                         \phi_{3}^{-1} : D_{1,3} \rightarrow  \mathcal{S}^{2}_{ \{x^{2} > 0\} }  \\
                                         (x^{1}, x^{3}) \mapsto \big( x^{1}, \sqrt{1 - (x^{1})^{2} - (x^{3})^{2}}, x^{3} \big)
                                         \end{array}
\end{array}
\end{equation}
If we choose now a point $P$ of the sphere, lying in the intersection of the above domains: $P \in \mathcal{S}^{2}_{ \{x^{3} > 0\} } \bigcap \mathcal{S}^{2}_{ \{x^{2} > 0\} }$ (in other words we have: $x_{P}^{2} > 0$ and $x_{P}^{3} > 0$) it is enough to note that the transition function
\begin{small}
\begin{equation}  \eqlabel{transfuncsphere}
\begin{array}{c}
\phi_{1} \circ \phi_{3}^{-1}: D_{1,3} \rightarrow D_{1,2}   \\
     \\
\phi_{1} \circ \phi_{3}^{-1}(x^{1}, x^{3}) = \phi_{1} \big( x^{1}, \sqrt{1 - (x^{1})^{2} - (x^{3})^{2}}, x^{3} \big) = \big( x^{1}, \sqrt{1 - (x^{1})^{2} - (x^{3})^{2}} \big)
\end{array}
\end{equation}
\end{small}
is $\mathcal{C}^{\infty}$-differentiable in its domain $D_{1,3}$. The same can be shown (in exactly the same way) for any other transition function defined by the maps $\phi_{i}$ for $i = 1, 2, ..., 6$. We have thus shown that the described homeomorphic vertical projections constitute a $\mathcal{C}^{\infty}$-atlas on $\mathcal{S}^{2}$

\underline{\textbf{Proof of (II):}}    \\
First we observe that the domains of the stereographic projections from the north \equref{sphere21} and the south pole \equref{sphere22} provide an open covering of the sphere. The inverse homeomorphisms are given by
\begin{equation}
\begin{array}{cccc}
\begin{array}{c}
\sigma_{1}^{-1} : \pi (x^{3} = 0) \rightarrow  \mathcal{S}^{2} - \{ 0, 0, 1 \}   \\
(\zeta, \eta) \mapsto (\frac{2 \zeta}{1 + \zeta^{2} + \eta^{2}}, \frac{2 \eta}{1 + \zeta^{2} + \eta^{2}}, - \frac{1 - \zeta^{2} - \eta^{2}}{1 + \zeta^{2} + \eta^{2}})
\end{array}    & & &  \begin{array}{c}
                      \sigma_{2}^{-1} : \pi (x^{3} = 0) \rightarrow  \mathcal{S}^{2} - \{ 0, 0, -1 \}    \\
                      (\zeta', \eta') \mapsto (\frac{2 \zeta'}{1 + \zeta'^{2} + \eta'^{2}}, \frac{2 \eta'}{1 + \zeta'^{2} + \eta'^{2}}, \frac{1 - \zeta'^{2} - \eta'^{2}}{1 + \zeta'^{2} + \eta'^{2}})
                      \end{array}
\end{array}
\end{equation}
In order that $\{ (\mathcal{S}^{2} - \{ 0,0,1 \}, \sigma_{1}), (\mathcal{S}^{2} - \{ 0,0, -1 \}, \sigma_{2}) \}$ constitutes a $\mathcal{C}^{\infty}$-atlas on the sphere, it is enough to note that the change of coordinates function (transition function)
\begin{equation}
\begin{array}{c}
\sigma_{1} \circ \sigma_{2}^{-1} : \mathbb{R}^{2} - \{(0,0)\} \rightarrow \mathbb{R}^{2} - \{(0,0)\} \\
   \\
\sigma_{1} \circ \sigma_{2}^{-1}(\zeta', \eta') = \sigma_{1} \big( \frac{2 \zeta'}{1 + \zeta'^{2} + \eta'^{2}}, \frac{2 \eta'}{1 + \zeta'^{2} + \eta'^{2}}, \frac{1 - \zeta'^{2} - \eta'^{2}}{1 + \zeta'^{2} + \eta'^{2}} \big) =  \\
    \\
\Big( \frac{\frac{2 \zeta'}{1 + \zeta'^{2} + \eta'^{2}}}{1 - \frac{1 - \zeta'^{2} - \eta'^{2}}{1 + \zeta'^{2} + \eta'^{2}}}, \frac{\frac{2 \eta'}{1 + \zeta'^{2} + \eta'^{2}}}{1 - \frac{1 - \zeta'^{2} - \eta'^{2}}{1 + \zeta'^{2} + \eta'^{2}}} \Big) = \big( \frac{\zeta'}{\zeta'^{2} + \eta'^{2}}, \frac{\eta'}{\zeta'^{2} + \eta'^{2}} \big)
\end{array}
\end{equation}
or equivalently
$$
\begin{array}{ccccc}
\zeta (\zeta', \eta') = \frac{\zeta'}{\zeta'^{2} + \eta'^{2}} &  \  & and: &  \  & \eta (\zeta', \eta') = \frac{\eta'}{\zeta'^{2} + \eta'^{2}}
\end{array}
$$
is $\mathcal{C}^{\infty}$-differentiable in its domain $\mathbb{R}^{2} - \{(0,0)\}$. Since this is evident from the above expressions, the proof is complete.

\underline{\textbf{Proof of (III):}}    \\
In order to show that the $\mathcal{C}^{\infty}$-atlases constructed above are equivalent (i.e. they determine the same $\mathcal{C}^{\infty}$-structure on $\mathcal{S}^{2}$) it suffices to show that any two of their local charts are $\mathcal{C}^{\infty}$-compatible. So if we pick for example the local charts $(\mathcal{S}^{2}_{ \{x^{2} > 0\} }, \phi_{3})$ and $(\mathcal{S}^{2} - \{ 0,0,1 \}, \sigma_{1})$ from the $\mathcal{C}^{\infty}$-atlases constructed in $\mathbf{(I)}$ and $\mathbf{(II)}$ respectively, we have that
$$
\sigma_{1} \Big( \mathcal{S}^{2}_{ \{x^{2} > 0\} } \bigcap \mathcal{S}^{2} - \{ 0,0,1 \} \Big) = \sigma_{1} \big( \mathcal{S}^{2}_{ \{x^{2} > 0\} } \big) = \mathbb{R}^{2}_{ \{x^{2} > 0\} }
$$
where $\mathbb{R}^{2}_{ \{x^{2} > 0\} }$ is the open half-plane. We thus get for the transition function (change of coordinates) between the local charts of the two specified atlases
\begin{equation}
\begin{array}{c}
\phi_{3} \circ \sigma_{1}^{-1}: \mathbb{R}^{2}_{ \{x^{2} > 0\} } \rightarrow D_{1,3}    \\
    \\
\phi_{3} \circ \sigma_{1}^{-1} (\zeta, \eta) = \phi_{3}(\frac{2 \zeta}{1 + \zeta^{2} + \eta^{2}}, \frac{2 \eta}{1 + \zeta^{2} + \eta^{2}}, - \frac{1 - \zeta^{2} - \eta^{2}}{1 + \zeta^{2} + \eta^{2}}) = (\frac{2 \zeta}{1 + \zeta^{2} + \eta^{2}}, - \frac{1 - \zeta^{2} - \eta^{2}}{1 + \zeta^{2} + \eta^{2}})
\end{array}
\end{equation}
which is a $C^{\infty}$-differentiable function in its domain. Thus the proof is complete.
\end{example}

\begin{example}[\textsf{The \textbf{real Projective space} $\mathbf{P^{2}(\mathbb{R})}$ as a differentiable manifold}]  \  \\
Consider the Euclidean space with the origin excluded $\mathbb{R}^{3} - \{ (0,0,0) \}$ and define the following binary relation on it: If $\vec{x}, \vec{y} \in \mathbb{R}^{3} - \{ (0,0,0) \}$ then
\begin{equation}
\vec{x} \sim \vec{y} \Leftrightarrow \exists \ \lambda \in \mathbb{R}^{*} : \vec{y} = \lambda \vec{x}
\end{equation}
We can show that this is an equivalence relation on $\mathbb{R}^{3} - \{ (0,0,0) \}$. In other words two points of the $3d$-Euclidean space will be equivalent if and only if their corresponding position vectors are parallel. The set of equivalence classes
\begin{small}
\begin{equation}
\mathbf{P^{2}(\mathbb{R})} \equiv \mathbb{R}^{3} - \{ (0,0,0) \} \Big/ \sim \ = \big\{ [\vec{x}] \equiv (x^{1}:x^{2}:x^{3}) \ \diagup \ [\vec{x}]: equivalence \ class \ of \ \sim \big\}
\end{equation}
\end{small}
where $(x^{1}:x^{2}:x^{3}) \equiv [\vec{x}] = \big\{ \vec{y} \in \mathbb{R}^{3} - \{ (0,0,0) \} \diagup \exists \ \lambda \in \mathbb{R}^{*} : \vec{y} = \lambda \vec{x} \big\}$
will be called the (real) \textbf{projective plane}. It is easy to see that, as a set, the (real) projective plane $P^{2}(\mathbb{R})$ is identified with the set of all straight lines of $\mathbb{R}^{3}$ coming through the origin (with the origin excluded of course). Considering the canonical projection
\begin{equation}
\begin{array}{c}
\pi: \mathbb{R}^{3} - \{ (0,0,0) \} \rightarrow P^{2}(\mathbb{R})   \\
\vec{x} \equiv (x^{1}, x^{2}, x^{3}) \mapsto \pi(\vec{x}) = [\vec{x}] \equiv (x^{1}:x^{2}:x^{3})
\end{array}
\end{equation}
the projective plane acquires the induced topology from $\mathbb{R}^{3} - \{ (0,0,0) \}$ (through $\pi$). In this way $P^{2}(\mathbb{R})$ becomes a Hausdorff (why?) and connected (why?)  topological space. Thus the collection of the open sets of $P^{2}(\mathbb{R})$ will be
$$
\tau_{P^{2}(\mathbb{R})} = \big\{ U \ \diagup \ \pi_{-1}(U): open \ in \ \mathbb{R}^{3} - \{ (0,0,0) \} \big\}
$$
We are now going to specify a $C^{\infty}$-differentiable structure on $P^{2}(\mathbb{R})$: Consider the subsets $U_{i} = \{ [\vec{x}] \equiv (x^{1}:x^{2}:x^{3}) \ \diagup \ x^{i} \neq 0 \}$ for $i = 1, 2, 3$ of the projective plane. These are open in $P^{2}(\mathbb{R})$ and they apparently constitute an open covering of $P^{2}(\mathbb{R})$. Notice that, if $x^{1} \neq 0$, since $(x^{1}, x^{2}, x^{3}) \sim (1, \frac{x^{2}}{x^{1}}, \frac{x^{3}}{x^{1}})$ (in $\mathbb{R}^{3} - \{ (0,0,0) \}$) we have $(x^{1}:x^{2}:x^{3}) \equiv (1:\frac{x^{2}}{x^{1}}:\frac{x^{3}}{x^{1}}) \in U_{1} \subset P^{2}(\mathbb{R})$. The corresponding relations hold for $U_{2}, U_{3}$ also.
If we consider the mappings $\phi_{i}: U_{i} \rightarrow \mathbb{R}^{2}$ defined by
\begin{small}
\begin{equation}
\begin{array}{ccccc}
\phi_{1}\big( (1:\frac{x^{2}}{x^{1}}:\frac{x^{3}}{x^{1}}) \big) = (\frac{x^{2}}{x^{1}}, \frac{x^{3}}{x^{1}}),  & &
\phi_{2}\big( (\frac{x^{1}}{x^{2}}:1:\frac{x^{3}}{x^{2}}) \big) = (\frac{x^{1}}{x^{2}}, \frac{x^{3}}{x^{2}}),  & &
\phi_{3}\big( (\frac{x^{1}}{x^{3}}:\frac{x^{2}}{x^{3}}:1) \big) = (\frac{x^{1}}{x^{3}}, \frac{x^{2}}{x^{3}})
\end{array}
\end{equation}
\end{small}
then we can see that they are homeomorphisms and $\phi(U_{i}) = \mathbb{R}^{2}$. The geometrical interpretation of the $\phi_{i}$ maps ($i = 1, 2, 3$) can be deduced from the next figure
\begin{center}
\includegraphics[angle=270, scale=0.45]{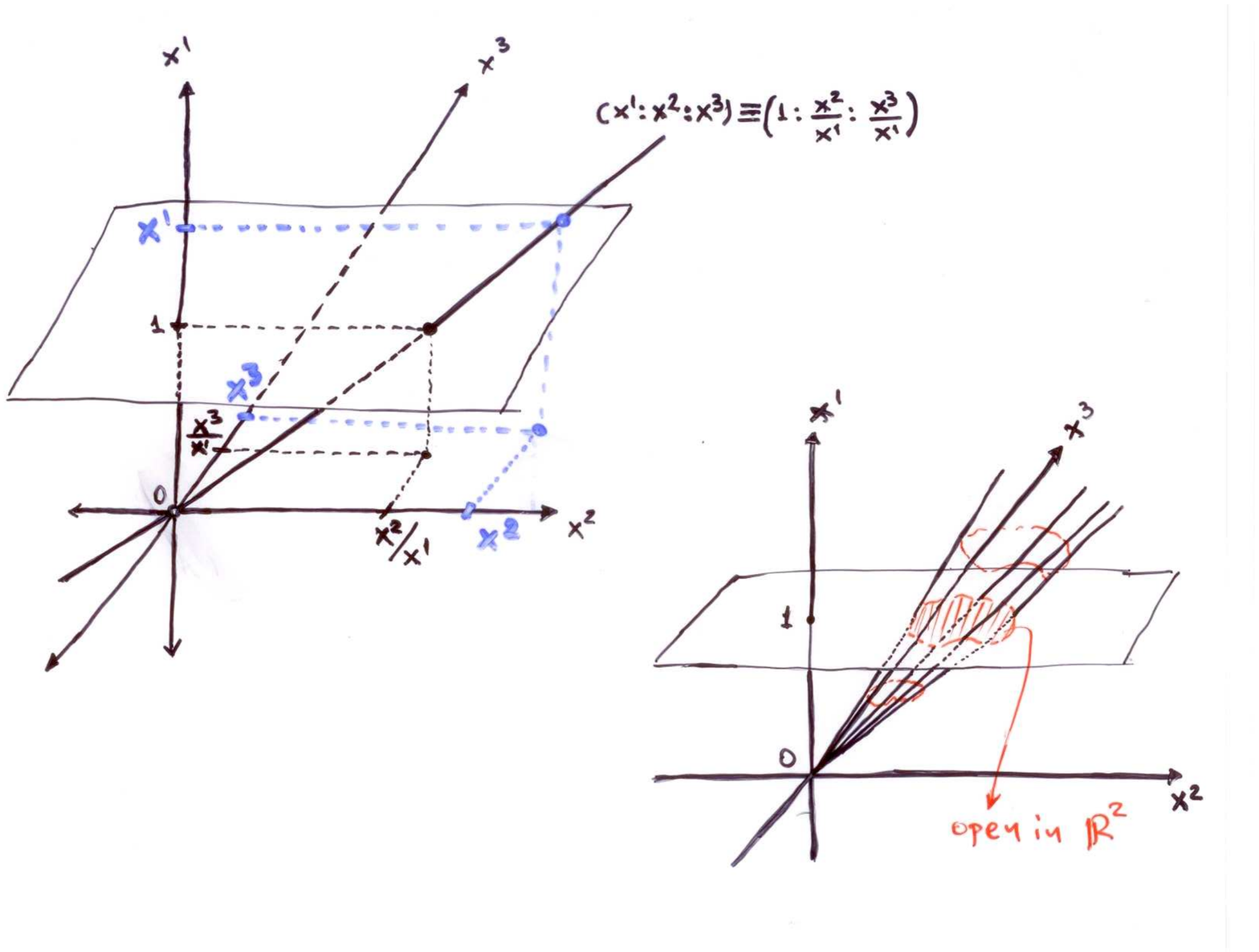}
\end{center}
For example, the $\phi_{1}$ map, maps the straight line $(x^{1}:x^{2}:x^{3}) \equiv (1:\frac{x^{2}}{x^{1}}:\frac{x^{3}}{x^{1}}) \in P^{2}(\mathbb{R})$ of the projective plane, to the point $(\frac{x^{2}}{x^{1}}, \frac{x^{3}}{x^{1}}) \in \mathbb{R}^{2}$ of the plane. But these are exactly the coordinates along $x^{2}$ and $x^{3}$ respectively, of the intersection point between the $(x^{1}:x^{2}:x^{3})$ straight line and the $x^{1} = 1$ plane.   \\
To verify that the $\phi_{i}$ mappings are indeed homeomorphisms, it suffices to construct the inverse mappings $\phi_{i}^{-1}: \mathbb{R}^{2} \rightarrow U_{i}$
\begin{small}
\begin{equation}
\begin{array}{ccccc}
\phi_{1}^{-1}\big( (x^{1}, x^{2}) \big) = (1:x^{1}:x^{2}),  & &
\phi_{2}^{-1}\big( (x^{1}, x^{2}) \big) = (x^{1}:1:x^{2}),  & &
\phi_{3}^{-1}\big( (x^{1}, x^{2}) \big) = (x^{1}:x^{2}:1)
\end{array}
\end{equation}
\end{small}
and show that both $\phi, \phi^{-1}$ are continuous (Exercise!). The geometrical interpretation of the inverse maps $\phi_{i}^{-1}$ for ($i = 1, 2, 3$) can again be deduced from the last figure: For example, the $\phi_{1}^{-1}$ map, maps the point $(x^{1}, x^{2}) \in \mathbb{R}^{2}$ to the straight line of $\mathbb{R}^{3} - \{ (0,0,0) \}$ constructed as follows: We consider the point $(x^{1}, x^{2})$ vertically ``raised'' to the $x^{1} = 1$ plane so its coordinates become $(1, x^{1}, x^{2})$ (in $\mathbb{R}^{3}$). Now the image of $(x^{1}, x^{2}) \in \mathbb{R}^{2}$ under $\phi_{1}^{-1}$, in other words the straight line $\phi_{1}^{-1}\big( (x^{1}, x^{2}) \big) = (1:x^{1}:x^{2})$ is exactly the straight line joining the origin with $(1, x^{1}, x^{2}) \in \mathbb{R}^{3}$.

Now it remains to check the transition functions (change of coordinates). Thus we have that if for example $(x^{1}, x^{2}) \in \phi_{3}(U_{1} \bigcap U_{3}) = \{ (x^{1}, x^{2}) \in \mathbb{R}^{2} \diagup x^{1} \neq 0 \}$ then
\begin{equation}
\phi_{1} \circ \phi_{3}^{-1}((x^{1}, x^{2})) = \phi_{1}((x^{1}:x^{2}:1)) = (\frac{x^{2}}{x^{1}}, \frac{1}{x^{1}})
\end{equation}
which is obviously a $C^{\infty}$-differentiable function in its domain $\phi_{3}(U_{1} \bigcap U_{3})$ ($\equiv$ the plane with the vertical $x_{2}$-axis removed). We have thus proved, that the projective plane is a $\mathcal{C}^{\infty}$-differentiable manifold whose dimension equals $2$.
\end{example}
The previous examples reveal some flavor of the wide variety of objects associated with the notion of differentiable manifolds. We saw that a differentiable manifold can be either some very familiar geometrical object like the Euclidean space, the graph of a continuous function in a Euclidean space, even a closed surface like a sphere (not constituting the graph of some function) or some quite unfamiliar (by intuition) geometrical object like the projective space or finally even some object with no obvious geometrical interpretation like the various matrix groups examined previously.
%\begin{scriptsize}
\begin{remark}[\textbf{Classification of $1$-manifolds}]
As it is customary in almost all branches of mathematics, the introduction of a novel abstraction or even a novel notion is almost automatically accompanied by attempts of classifying all possible examples associated with this notion. Far from trying to discuss classification results of differentiable manifold theory at this stage (which are an active and promising area of research in differential geometry nowadays), it is very instructive to underline an old result dealing with the \underline{classification of $1$-dimensional topological manifolds}: \textbf{A second countable $1$-dim. topological manifold, is homeomorphic to $\mathcal{S}^{1}$ if it is compact and to $\mathbb{R}$ if it is not}. In other words, the only second countable, $1$-dimensional, topological manifolds are the real line and the circle.
\end{remark}
%\end{scriptsize}
\textbf{\underline{Exercises:}}   \\
$\mathbf{3.}$ Show that any open, connected subset of a differentiable manifold is a differentiable manifold of the same dimension. (What if instead of an open subset we considered a closed subset of the initial manifold? Can you figure out where problems might arise?)   \\
$\mathbf{4.}$ If $\mathcal{N}$ and $\mathcal{M}$ are differentiable manifolds of dimensions $dim\mathcal{N} = n$ and $dim\mathcal{M} = m$ respectively, show that their Cartesian product $\mathcal{N} \times \mathcal{M}$ is a differentiable manifold of dimension $dim(\mathcal{N} \times \mathcal{M}) = n + m$    \\
$\mathbf{5.}$ Show that any regular surface $\overline{x} = \overline{x}(u,v)$ of $\mathbb{R}^{3}$ is a differentiable manifold whose dimension is $2$. \\
$\mathbf{6.}$ Show that the circle $\mathcal{S}^{1}$ is an one dimensional differentiable manifold   \\
$\mathbf{7.}$ Show that the $n$-torus $\mathcal{S}^{n} = \mathcal{S}^{1} \times \mathcal{S}^{1} \times ... \times \mathcal{S}^{1}$ is a diff. manif. with $dim\mathcal{S}^{n} = n$   \\
$\mathbf{8.}$ Show that the cylinder $\mathcal{S}^{1} \times \mathbb{R}$ is a diff. manif. whose dimension equals $2$.   \\
$\mathbf{9.}$ Compute the Jacobian for the change of coordinates (transition function) \equref{transfuncsphere}    \\

\section{Differentiable maps}

\begin{definition}[$\mathbf{\mathcal{C}^{r}}$\textbf{-differentiable functions} between manifolds] \delabel{smooth}
Let $\mathcal{N}$ and $\mathcal{M}$ be two $\mathcal{C}^{r}$-differentiable manifolds of dimensions $n$ and $m$ respectively. The map $f: \mathcal{N} \rightarrow \mathcal{M}$ will be said to be a $\mathbf{\mathcal{C}^{r}}$\textbf{-differentiable function at $P \in \mathcal{N}$} if there are local charts $(U_{P}, \phi)$ and $(V_{f(P)}, \psi)$ respectively, such that the real (vector) function
\begin{equation} \eqlabel{coordexproff}
\psi \circ f \circ \phi^{-1}: \phi \big( f^{-1}(V_{f(P)}) \cap U_{P}\big) \rightarrow \psi(V_{f(P)})
\end{equation}
is $\mathcal{C}^{r}$-differentiable at $\phi(P) \in \mathbb{R}$ (in the usual sense of calculus). \\
Notice that $\phi \big( f^{-1}(V_{f(P)} \cap U_{P})\big) \subseteq \mathbb{R}^{n}$ and $\psi(V_{f(P)}) \subseteq \mathbb{R}^{m}$. The situation is visually described in the following figure
\begin{center}
\includegraphics[angle=270, scale=0.5]{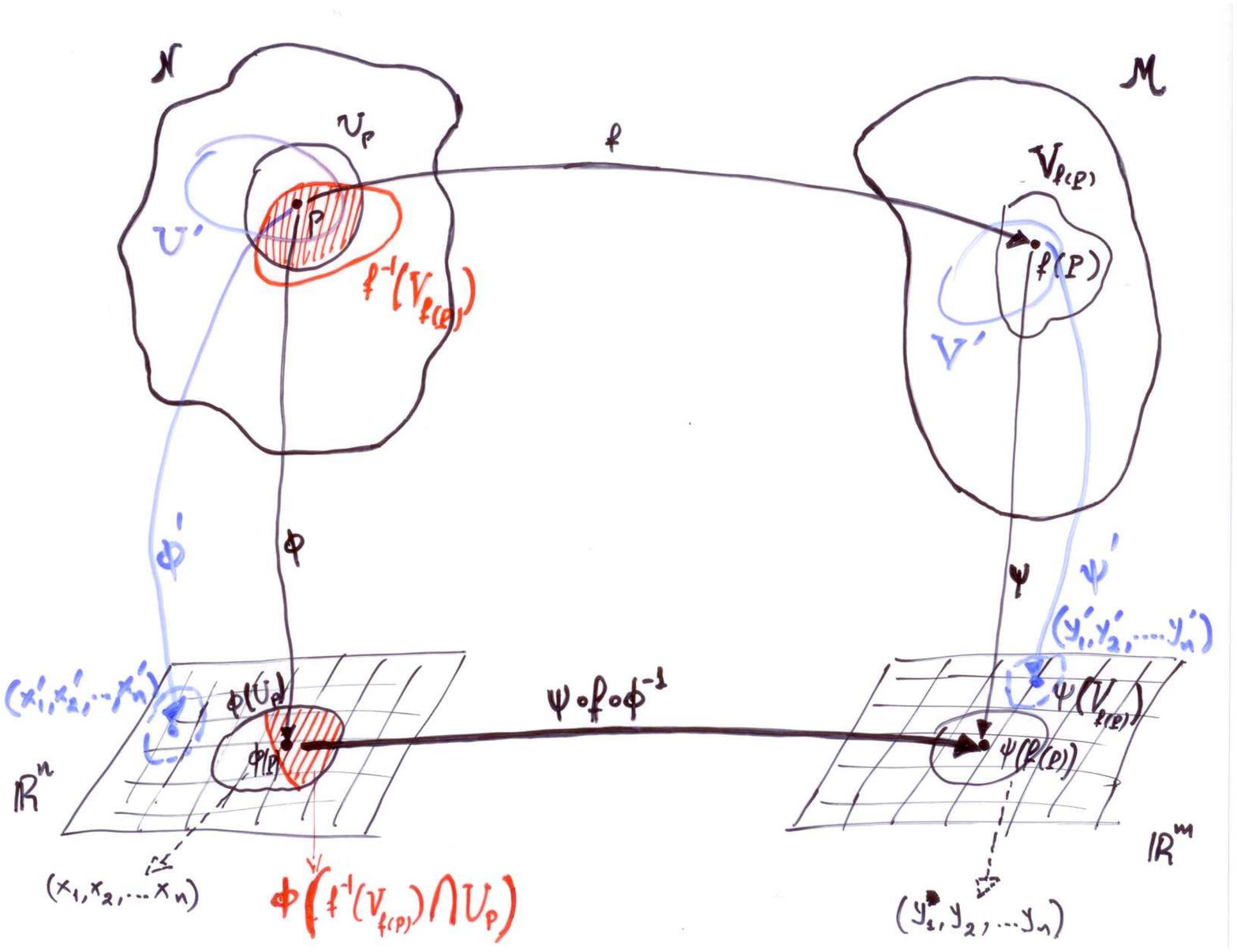}
\end{center}
The function $f: \mathcal{N} \rightarrow \mathcal{M}$ will be said to be a $\mathbf{\mathcal{C}^{r}}$\textbf{-differentiable function} if it is $\mathcal{C}^{r}$-differentiable for every $P \in \mathcal{N}$.
\end{definition}
The following lemma verifies the validity and the usefulness of the previous definition
\begin{lemma}
The definition of $\mathcal{C}^{r}$-differentiable function given above is independent of the choice of local charts
$(U_{P}, \phi)$ and $(V_{f(P)}, \psi)$.
\end{lemma}
\begin{proof}
If we had chosen two different local charts $(U^{'}_{P}, \phi^{'})$ and $(V^{'}_{f(P)}, \psi^{'})$ then
$$
\psi^{'} \circ f \circ \phi^{' -1} =  (\psi^{'} \circ \psi^{-1}) \circ (\psi \circ f \circ \phi^{-1}) \circ (\phi \circ \phi^{' -1})
$$
But all the (real) functions in the rhs of the above are $\mathcal{C}^{r}$-differentiable (in the usual sense of calculus) by the definition of the notion of $\mathcal{C}^{r}$-differentiable manifold.

Consequently, we have shown that if $f: \mathcal{N} \rightarrow \mathcal{M}$ is a $\mathcal{C}^{r}$-differentiable function for one choice of the local charts $(U_{P}, \phi)$ and $(V_{f(P)}, \psi)$ then it will be so for any other choice of local charts.
\end{proof}
The real (vector) function $\psi \circ f \circ \phi^{-1}(x_{1}, x_{2}, ..., x_{n}) = (y_{1}, y_{2}, ..., y_{n})$ can be written equivalently
$$
y_{i} = y_{i}(x_{1}, x_{2}, ..., x_{n})
$$
for all values $i = 1, 2, ..., n$.
\begin{example}[\textsf{Real functions on a manifold}]
If $\mathcal{M} = \mathbb{R}$ then the real function $f: \mathcal{N} \rightarrow \mathbb{R}$ will be $\mathcal{C}^{r}$-differentiable at $P \in \mathcal{N}$ iff $\exists$ a local chart $(U_{P}, \phi)$ such that $f \circ \phi^{-1}$ is $\mathcal{C}^{r}$-differentiable at $\phi(P) \in \mathbb{R}$ (in the usual sense of calculus).

We are going to use the following notations from now on
\begin{equation}
D^{r}(\mathcal{N}, P) = \{ f \diagup f: \mathcal{N} \rightarrow \mathbb{R} \ and \ \mathcal{C}^{r}-differentiable \ at \ P \in \mathcal{N}\}
\end{equation}
\begin{equation} \eqlabel{algdiffunctonN}
D^{r}(\mathcal{N}) = \{ f \diagup f: \mathcal{N} \rightarrow \mathbb{R} \ and \ \mathcal{C}^{r}-differentiable \ on \ \mathcal{N}\}
\end{equation}
\end{example}
%\begin{scriptsize}
\begin{remark}[\textbf{Algebra of differentiable functions on a manifold}] \relabel{Ralgebra}
If $\mathcal{V}_{n}$ is an $n$-dim., $\mathcal{C}^{r}$-differentiable manifold and $P \in \mathcal{V}_{n}$ is any point of the manifold, then for any pair of functions $f,g \in D^{r}(\mathcal{V}_{n}, P)$ we can define their sum $f+g$ by $(f + g)(P) = f(P) + g(P)$ their product $f \cdot g$ by $(f \cdot g)(P) = f(P)g(P)$ and the scalar product $\lambda f$ by $(\lambda f)(P) = \lambda f(P)$ (for any real number $\lambda$). We can show that \underline{under these operations $D^{r}(\mathcal{V}_{n}, P)$ becomes an $\mathbb{R}$-algebra.}
\end{remark}
%\end{scriptsize}
\begin{example}[$\mathcal{C}^{r}$-\textsf{differentiable curve on a manifold}] \exlabel{smoothcurve}
Let the $1$-dim. $\mathcal{C}^{r}$-manifold $\mathcal{N} = (0,1)$. We will call a \textbf{$\mathcal{C}^{r}$-differentiable curve on $\mathcal{M}$} the $\mathcal{C}^{r}$-differentiable function $\gamma: (0,1) \rightarrow \mathcal{M}$. In terms of coordinates we may write
\begin{equation}
x_{i}(t) \equiv \pi_{i} \circ \psi \circ \gamma (t)
\end{equation}
where $i = 1, 2, ..., n$ and $t \in (0,1)$.
\end{example}
\begin{remark}
Note that according to the definition of smooth ($\mathcal{C}^{r}$) differentiable curve given in \exref{smoothcurve} a differentiable curve on a manifold may admit cusps as is the case in the following example:
\begin{center}
\includegraphics[scale=0.5]{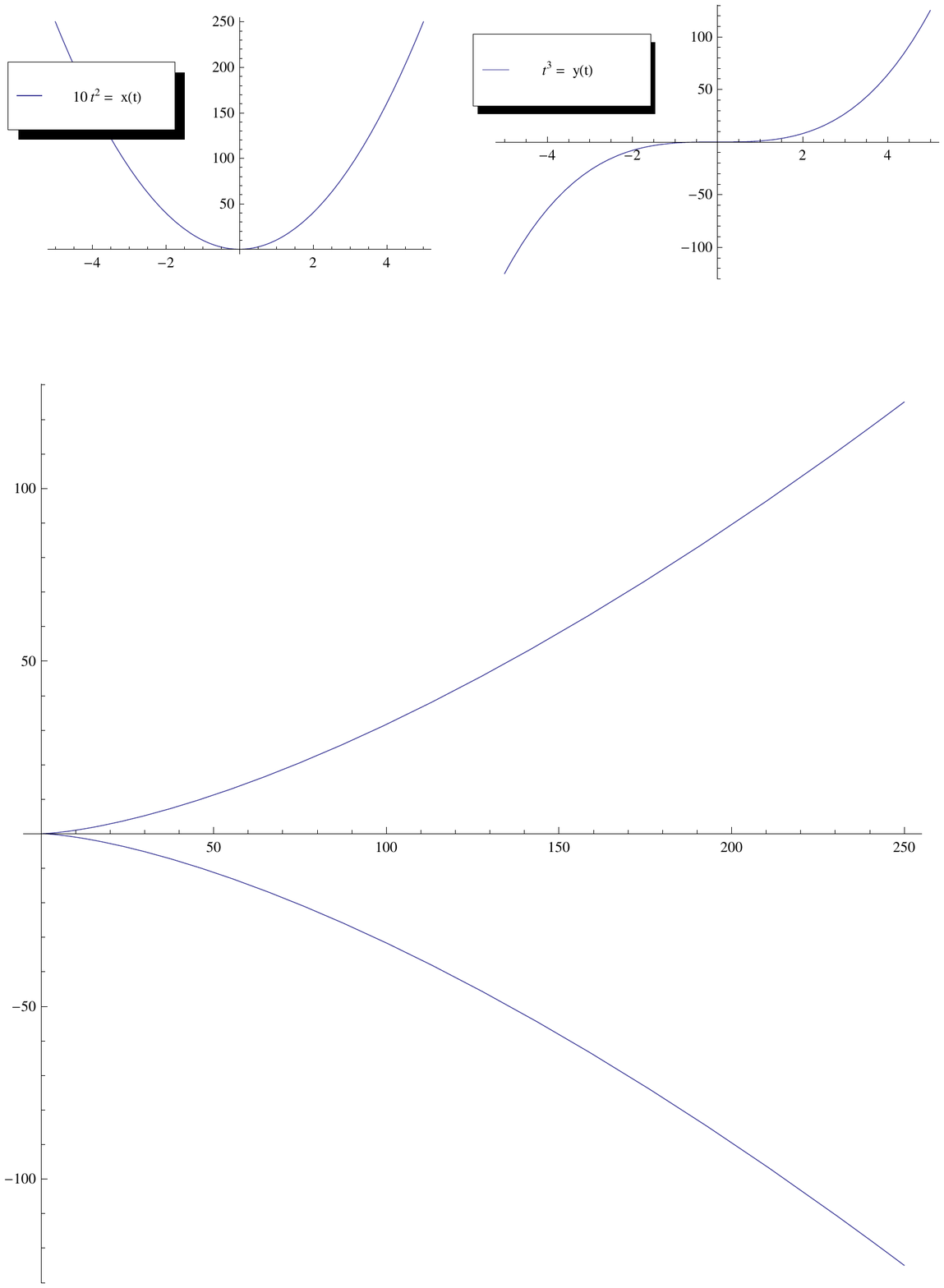}
\end{center}
\end{remark}
It is obvious that the above definition can be readily generalized to the definition of the \textbf{$\mathcal{C}^{r}$-differentiable (hyper)surface on $\mathcal{M}$}.

Let $f: \mathcal{N} \rightarrow \mathcal{M}$ be a homeomorphism. If $f$ is also $\mathcal{C}^{r}$-differentiable, its inverse function $f^{-1}: \mathcal{M} \rightarrow \mathcal{N}$ need not be also $\mathcal{C}^{r}$-differentiable.
\begin{definition}[$\mathbf{\mathcal{C}^{r}}$-\textbf{Diffeomorphism}]
Given the $\mathcal{C}^{r}$-differentiable function $f: \mathcal{N} \rightarrow \mathcal{M}$, if the inverse map $f^{-1}:\mathcal{M} \rightarrow \mathcal{N}$ is also $\mathcal{C}^{r}$-differentiable function then $f$ will be called a $\mathbf{\mathcal{C}^{r}}$-\textbf{diffeomorphism}.
\end{definition}
Diffeomorphisms of smooth manifolds play the same role as homeomorphisms do for topological spaces. If $f: \mathcal{N} \rightarrow \mathcal{M}$ is a diffeomorphism then the manifolds $\mathcal{N}$ and $\mathcal{M}$ will be called diffeomorphic. The set of all manifolds is subdivided into non-intersecting classes of pairwise diffeomorphic manifolds. Any general property of smooth manifolds or smooth mappings on manifolds can be transferred to any other diffeomorphic manifold. Therefore, we shall not distinguish between diffeomorphic manifolds.
%\begin{scriptsize}
\begin{remark}[\textbf{Group of diffeomorphisms on a manifold}] \relabel{GroupofDiffeom}
The set $\underline{\mathcal{D}iff(\mathcal{M})}$ of all diffeomorphisms $f: \mathcal{M} \rightarrow \mathcal{M}$ can be easily shown to be a group. It is called the \underline{group of diffeomorphisms on $\mathcal{M}$}. This group plays an important role in various branches of theoretical physics.
\end{remark}
%\end{scriptsize}

\textbf{\underline{Exercises:}}   \\
$\mathbf{1.}$ If $\mathcal{N}$ is an $n$-dim., $\mathcal{C}^{r}$-differentiable manifold then prove that the coordinate homeomorphisms $\phi_{i}: U_{i} \rightarrow \phi(U_{i}) \subseteq \mathbb{R}^{n}$ are $\mathcal{C}^{r}$-diffeomorphisms. \\
$\mathbf{2.}$ Prove that the domain $\phi \big( f^{-1}(V_{f(P)}) \cap U_{P}\big)$ of definition of $\psi \circ f \circ \phi^{-1}$ in \equref{coordexproff} is an open subset of $\mathbb{R}^{n}$.   \\
$\mathbf{3.}$  Prove that \deref{smooth} implies that the function $f: \mathcal{N} \rightarrow \mathcal{M}$ has to be a continuous function.   \\
$\mathbf{4.}$ Prove the assertion of \reref{Ralgebra}, i.e. that $D^{r}(\mathcal{V}_{n}, P)$ becomes an $\mathbb{R}$-algebra.   \\
$\mathbf{5.}$ Prove the assertion of \reref{GroupofDiffeom}, i.e. that $\mathcal{D}iff(\mathcal{M})$ becomes an $\mathbb{R}$-algebra. \\

\section{Tangency: spaces, bundles and tensor fields on manifolds}

\subsection{The Tangent space}

%\subsubsection{Abstract (algebraic) definition}

\paragraph{$\blacksquare$ Construction of the tangent space:}

\begin{definition}[\textbf{Tangent vector} on  point $P$ of a manifold] \delabel{tangvect}
Let $\mathcal{V}_{n}$ be an $n$-dimensional, $\mathcal{C}^{r}$-differentiable manifold and let $P \in \mathcal{V}_{n}$ be any point on it. We will call a \textbf{tangent vector} of the manifold $\mathcal{V}_{n}$ at the point $P$, \underline{any linear map} \underline{$X_{P}: D^{r}(\mathcal{V}_{n}, P) \rightarrow \mathbb{R}$} which furthermore \underline{satisfies the following condition}: \\
If $g_{1}, g_{2}, ..., g_{k}$ is a (finite) collection of $\mathcal{C}^{r}$-differentiable functions of $D^{r}(\mathcal{V}_{n}, P)$ and $f$ is any $\mathcal{C}^{r}$-differentiable function of the $g_{i}$ $(i = 1, 2, ..., k)$ then
\begin{equation} \eqlabel{abstrdeftangect}
X_{P} \big( f(g_{1}, g_{2}, ..., g_{k}) \big) = \sum_{i=1}^{k} \Big( \frac{\partial f}{\partial g_{i}} \Big)_{P} X_{P}(g_{i})
\end{equation}
\end{definition}
$\bullet$ Thus, according to the above definition, a tangent vector at a point $P$ of a manifold $\mathcal{V}_{n}$ is a linear functional of $D^{r}(\mathcal{V}_{n}, P)$ satisfying the chain rule of differentiation \equref{abstrdeftangect}.

$\bullet$ Note that in the above definition we use the notation
\begin{equation} \eqlabel{differonmanif}
\Big( \frac{\partial f(g_{1}, g_{2}, ..., g_{k})}{\partial g_{i}} \Big)_{P} \stackrel{def.}{\equiv} \Big( \frac{\partial f \big( g_{1} \circ \phi^{-1}, g_{2} \circ \phi^{-1}, ..., g_{k} \circ \phi^{-1} \big)}{ \partial \big( g_{i} \circ \phi^{-1} \big)} \Big)_{\phi(P)}
\end{equation}

Let us now apply \equref{abstrdeftangect} of \deref{tangvect} for the case $\underline{k=2}$: Let $g_{1}, g_{2} \in D^{r}(\mathcal{V}_{n}, P)$ and let $f(g_{1}, g_{2}) = g_{1}g_{2}$. Applying the above we get
\begin{equation} \eqlabel{Leibniz}
\begin{array}{c}
X_{P}(f(g_{1}, g_{2})) = \sum_{i=1}^{2} \Big( \frac{\partial f}{\partial g_{i}} \Big)_{P} X_{P}(g_{i}) =
\Big( \frac{\partial (g_{1}g_{2})}{\partial g_{1}} \Big)_{P}X_{P}(g_{1}) + \Big( \frac{\partial (g_{1}g_{2})}{\partial g_{2}} \Big)_{P}X_{P}(g_{2}) \Rightarrow  \\
   \\
\Rightarrow X_{P}(g_{1}g_{2}) = X_{P}(g_{1}) g_{2}(P) + X_{P}(g_{2}) g_{1}(P)
\end{array}
\end{equation}
We have thus shown the following
\begin{corollary}
Any tangent vector of $\mathcal{V}_{n}$ at the point $P$ satisfies the Leibniz rule. Consequently, any tangent vector of $\mathcal{V}_{n}$ at the point $P$ constitutes a differentiation on the $\mathbb{R}$-algebra $D^{r}(\mathcal{V}_{n}, P)$.
\end{corollary}
%\begin{scriptsize}
\begin{remark}
We have already shown that \underline{\equref{abstrdeftangect} implies \equref{Leibniz}}. It can be shown that \underline{the inverse implication is valid only in the case that the manifold $\mathcal{V}_{n}$ is $\mathcal{C}^{\infty}$ or $\mathcal{C}^{\omega}$}
\end{remark}
%\end{scriptsize}
Let us now denote by $T_{P}(\mathcal{V}_{n})$ the set of all tangent vectors of the manifold $\mathcal{V}_{n}$ at the point $P$. Based on the above it is easy to put a vector space structure on $T_{P}(\mathcal{V}_{n})$: If $X_{P}, Y_{P}$ are any elements of $T_{P}(\mathcal{V}_{n})$ then we define addition and scalar multiplication by
\begin{itemize}
\item[$\mathbf{-}$] $(X_{P} + Y_{P})(f) = X_{P}(f) + Y_{P}(f)$
\item[$\mathbf{-}$] $(\lambda X_{P})(f) = \lambda X_{P}(f)$
\end{itemize}
for any real number $\lambda$ and any $f \in D^{r}(\mathcal{V}_{n}, P)$. It is easy to show that $T_{P}(\mathcal{V}_{n})$ equipped with the above operations becomes a vector space.
\begin{definition}[\textbf{Tangent space} at a point $P$ of the manifold]
The real vector space $\mathbf{T_{P}(\mathcal{V}_{n})}$ will be called the \textbf{tangent space of the manifold at the point} $\mathbf{P \in \mathcal{V}_{n}}$.
\end{definition}

\paragraph{$\blacksquare$ Construction of basis in the tangent space $T_{P}(\mathcal{V}_{n})$:}

Let us now proceed in computing a basis of the vector space $T_{P}(\mathcal{V}_{n})$:

$\blacktriangleright$ In order to compute a basis of $T_{P}(\mathcal{V}_{n})$, let us choose a random vector $X_{P} \in T_{P}(\mathcal{V}_{n})$. Thus $X_{P}: D^{r}(\mathcal{V}_{n}, P) \rightarrow \mathbb{R}$ is a linear functional satisfying the chain rule of differentiation \equref{abstrdeftangect}. In order to express $X_{P}$ as a linear combination of other vectors let us choose, in place of the arbitrary functions $g_{i}$ of the \deref{tangvect} the \emph{coordinate projection functions} $g_{i} = x^{i} \equiv \pi^{i} \circ \phi$ with respect to some chart $(U_{P}, \phi)$ which contains the point $P$. It is obvious at this point, that we could have used any other chart $(V_{P}, \psi)$ containing the point $P$ and its corresponding coordinate projection functions $x^{' i} \equiv \pi^{i} \circ \psi$ (in the place of the $g_{i}$ functions) These functions are visualized in the figure that follows:
\begin{center}
\includegraphics[angle=270, scale=0.45]{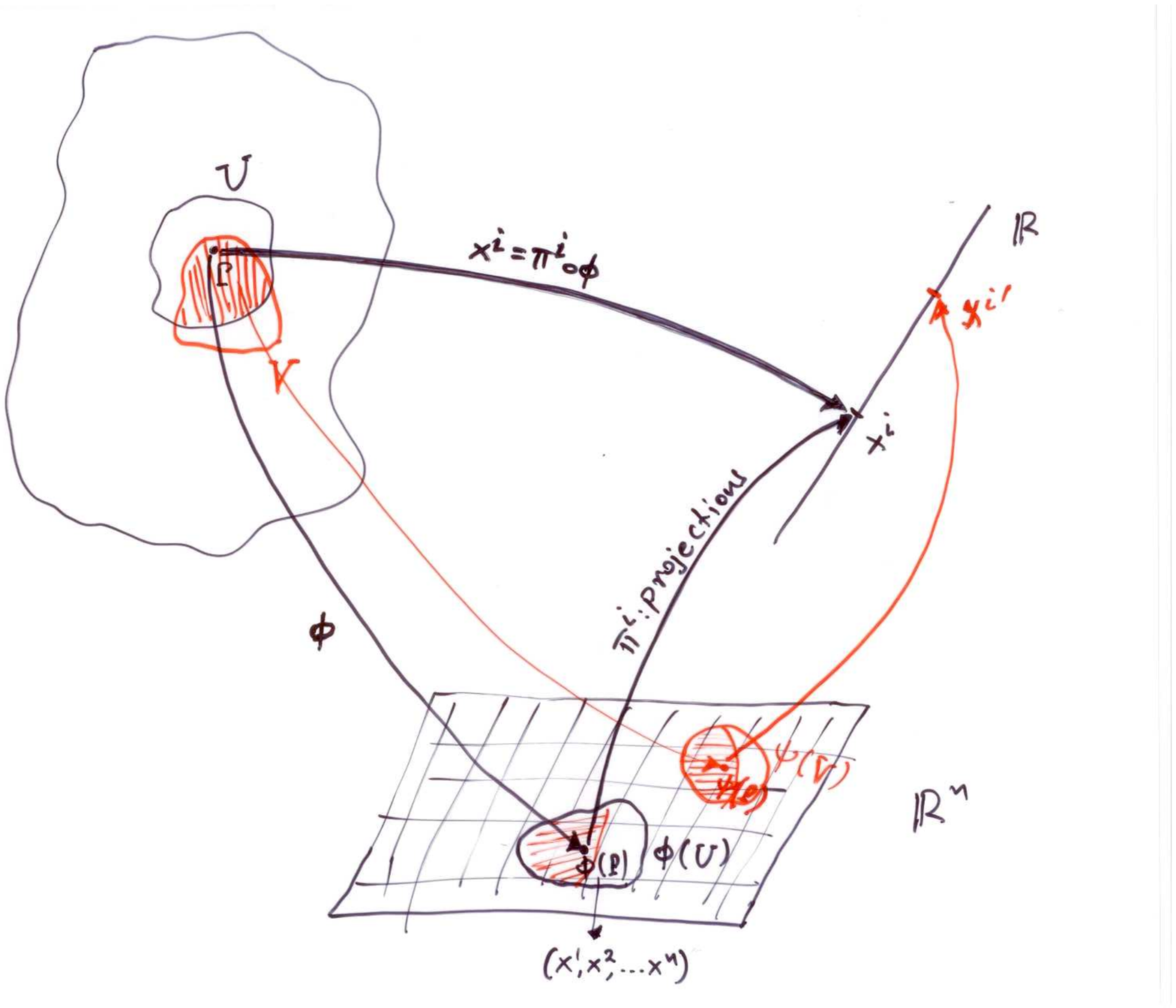}
\end{center}
Now we get
\begin{equation} \eqlabel{Xpaslinearcomb}
X_{P} \big( f(x^{1}, x^{2}, ..., x^{k}) \big) = \sum_{i=1}^{n} \Big( \frac{\partial f}{\partial x^{i}} \Big)_{P} X_{P}(x^{i}) = \sum_{i=1}^{n} \Big( \frac{\partial f}{\partial x^{i}} \Big)_{P} X_{P}^{i}
\end{equation}
where $X_{P}^{i} \equiv X_{P}(x^{i}) \in \mathbb{R}$ (by definition).

$\blacktriangleright$ Next we consider the linear mappings $e_{i} \equiv \frac{\partial}{\partial x^{i}} \big|_{P}: D^{r}(\mathcal{V}_{n}, P) \rightarrow \mathbb{R}$ (for $i = 1, 2, ..., n$) defined by
\begin{equation} \eqlabel{basisofTpwithresptoachart}
e_{i}(f) \equiv \frac{\partial}{\partial x^{i}} \Big|_{P}(f) = \Big( \frac{\partial f}{\partial x^{i}} \Big)_{P}
\end{equation}
for any $f \in D^{r}(\mathcal{V}_{n}, P)$. We immediately get the following
\begin{lemma} \lelabel{transferofdiffer}
We have the following
$$
e_{i}(f) \equiv \frac{\partial}{\partial x^{i}} \Big|_{P}(f) = \Big( \frac{\partial f}{\partial x^{i}} \Big)_{P} = D_{i}(f \circ \phi^{-1}) \Big|_{\phi(P)}
$$
where $D_{i}$ denotes the usual partial derivative with respect to the $x^{i}$ coordinate of $\mathbb{R}^{n}$.
\end{lemma}
\begin{proof}
This is a straightforward application of equation \equref{differonmanif}.
\end{proof}
\begin{lemma}
The $e_{i} \equiv \frac{\partial}{\partial x^{i}} \big|_{P}$ linear mappings (for $i = 1, 2, ..., n$) defined by \equref{basisofTpwithresptoachart} are tangent vectors of the manifold $\mathcal{V}_{n}$ at the point $P \in \mathcal{V}_{n}$ i.e. they are elements of $T_{P}(\mathcal{V}_{n})$.
\end{lemma}
\begin{proof}
Let $g_{1}, g_{2}, ..., g_{k} \in D^{r}(\mathcal{V}_{n}, P)$ and let $f$ be $\mathcal{C}^{r}$-function of the $g_{i}$ $(i = 1, 2, ..., k)$ functions. Then, applying \equref{basisofTpwithresptoachart} we get
\begin{small}
\begin{equation}
e_{i} \Big( f(g_{1}, g_{2}, ..., g_{k}) \Big) = \Big( \frac{\partial f(g_{1}, g_{2}, ..., g_{k})}{\partial x^{i}} \Big)_{P} = \sum_{r=1}^{k} \Big( \frac{\partial f}{\partial g_{r}} \Big)_{P} \Big( \frac{\partial g_{r}}{\partial x^{i}} \Big)_{P} = \sum_{r=1}^{k} \Big( \frac{\partial f}{\partial g_{r}} \Big)_{P} e_{i}(g_{r})
\end{equation}
\end{small}
which proves that the $e_{i}$ functions defined by \equref{basisofTpwithresptoachart} are tangent vectors of the manifold $\mathcal{V}_{n}$ at the point $P$, in the sense of the \deref{tangvect}.
\end{proof}

$\blacktriangleright$ Combining equations \equref{Xpaslinearcomb}, \equref{basisofTpwithresptoachart} we arrive at the following relation ($X_{P}$ is an arbitrary vector of $T_{P}(\mathcal{V}_{n})$ and $X_{P}^{i} \equiv X_{P}(x^{i}) \in \mathbb{R}$ by def.)
\begin{equation}
X_{P} \big( f(x^{1}, x^{2}, ..., x^{n}) \big) = \sum_{i=1}^{n} X_{P}^{i} e_{i}\big( f(x^{1}, x^{2}, ..., x^{n}) \big) \equiv \sum_{i=1}^{n} X_{P}^{i} \frac{\partial}{\partial x^{i}} \Big|_{P} \big( f(x^{1}, x^{2}, ..., x^{n}) \big)
\end{equation}
or:
\begin{equation} \eqlabel{lincombbasveclccoordsyst}
X_{P} = \sum_{i=1}^{n} X_{P}^{i} e_{i} \equiv \sum_{i=1}^{n} X_{P}^{i} \frac{\partial}{\partial x^{i}} \Big|_{P}
\end{equation}
which proves that \underline{the $e_{i} \equiv \frac{\partial}{\partial x^{i}} \big|_{P}$ span the space $T_{P}(\mathcal{V}_{n})$} ($i = 1, 2, ..., n$).

$\blacktriangleright$ Finally, we are now going to prove that the $e_{i} \equiv \frac{\partial}{\partial x^{i}} \big|_{P}$ vectors are linearly independent in $T_{P}(\mathcal{V}_{n})$. To see this let us choose the real numbers $\lambda^{i}$ ($i = 1, 2, ..., n$) such that $\lambda^{i} e_{i} = 0$. This is equivalent to saying that the linear map $\lambda^{i} e_{i}: D^{r}(\mathcal{V}_{n}, P) \rightarrow \mathbb{R}$ is identically zero \footnote{i.e. $\lambda^{i} e_{i}$ is the zero map} for any $f \in D^{r}(\mathcal{V}_{n}, P)$. Thus, if in place of $f$ we place the coordinate projection functions $x^{j} \equiv \pi^{j} \circ \phi$ ($j = 1, 2, ..., n$) with respect to some chart $(U_{P}, \phi)$ which contains the point $P$, we get \footnote{The Einstein summation convention is employed in \equref{linindepe}}
\begin{equation} \eqlabel{linindepe}
\lambda^{i} e_{i} (x^{j}) \equiv \lambda^{i} \frac{\partial}{\partial x^{i}} \Big|_{P} (x^{j}) = \lambda^{i} \Big( \frac{\partial x^{j}}{\partial x^{i}} \Big)_{P} = \lambda^{i} \delta_{i}^{j} = \lambda^{j} = 0
\end{equation}
for any value $j = 1, 2, ..., n$, because (applying \equref{differonmanif} or equivalently applying \leref{transferofdiffer})
\begin{small}
$$
\frac{\partial}{\partial x^{i}} \Big|_{P} (x^{j}) = \Big( \frac{\partial x^{j}}{\partial x^{i}} \Big)_{P} = D_{i}(x^{j} \circ \phi^{-1}) \Big|_{\phi(P)} = D_{i}(\pi^{j} \circ \phi \circ \phi^{-1}) \Big|_{\phi(P)} = D_{i}(\pi^{j}) \Big|_{\phi(P)} = \delta_{i}^{j}
$$ \end{small}
where $D_{i}$ denotes the usual partial derivative with respect to the $x^{i}$ coordinate of $\mathbb{R}^{n}$. We thus have shown that \underline{the $e_{i} \equiv \frac{\partial}{\partial x^{i}} \big|_{P}$ vectors are linearly independent in $T_{P}(\mathcal{V}_{n})$} ($i = 1, 2, ..., n$).

Finally, we have proved the following theorem
\begin{theorem}
If $\mathcal{V}_{n}$ is an $n$-dimensional, $\mathcal{C}^{r}$-differentiable manifold and $P \in \mathcal{V}_{n}$ any point on it, then the set $T_{P}(\mathcal{V}_{n})$ of all tangent vectors of the manifold $\mathcal{V}_{n}$ at the point $P$ becomes a (real) vector space whose dimension is given by
$$
dim T_{P}(\mathcal{V}_{n}) = dim \mathcal{V}_{n} = n
$$
Furthermore, if $(U_{P}, \phi)$ is a local chart containing $P$ (i.e. a local coordinate system in a neighborhood of $P$) then
\begin{equation} \eqlabel{basistangsp}
\Big\{ \frac{\partial}{\partial x^{1}} \Big|_{P}, \frac{\partial}{\partial x^{2}} \Big|_{P}, ..., \frac{\partial}{\partial x^{n}} \Big|_{P} \Big\} \equiv \Big\{ \frac{\partial}{\partial x^{1}}, \frac{\partial}{\partial x^{2}}, ..., \frac{\partial}{\partial x^{n}} \Big\}_{P}
\end{equation}
is the \textbf{natural basis of the tangent space} $\mathbf{T_{P}(\mathcal{V}_{n})}$ at the point $P$ \textbf{with respect to the local chart} $\mathbf{(U_{P}, \phi)}$.
\end{theorem}

\paragraph{$\blacksquare$ Change of basis in the tangent space $T_{P}(\mathcal{V}_{n})$:}

If in place of the local chart $(U_{P}, \phi)$ we use another \footnote{Both local charts $(U_{P}, \phi)$, $(V_{P}, \psi)$ must contain $P$ in order for our discussion to be meaningful} local chart $(V_{P}, \psi)$ whose local coordinate system around $P$ is denoted by $(x^{1 '}, x^{2 '}, ..., x^{n '})$ then we would get another basis od the tangent space $T_{P}(\mathcal{V}_{n})$ given by
\begin{equation} \eqlabel{altbasistangsp}
\Big\{ \frac{\partial}{\partial x^{1 '}} \Big|_{P}, \frac{\partial}{\partial x^{2 '}} \Big|_{P}, ..., \frac{\partial}{\partial x^{n '}} \Big|_{P} \Big\} \equiv \Big\{ \frac{\partial}{\partial x^{1 '}}, \frac{\partial}{\partial x^{2 '}}, ..., \frac{\partial}{\partial x^{n '}} \Big\}_{P}
\end{equation}
Our next task will be to determine the transformation law connecting the two different bases \equref{basistangsp}, \equref{altbasistangsp}: Let us consider a linear transformation between these two bases, given by \footnote{Einstein summation convention will be employed from now on without further mentioning so}
\begin{equation} \eqlabel{transflaw}
\frac{\partial}{\partial x^{i '}} \Big|_{P} = \lambda_{i '}^{j} \frac{\partial}{\partial x^{j}} \Big|_{P}
\end{equation}
where $\lambda_{i '}^{j}$ are real numbers to be determined. Since both the lhs and the rhs of \equref{transflaw} are linear maps from $D^{r}(\mathcal{V}_{n}, P)$ to $\mathbb{R}$, \equref{transflaw} must be valid for any $f \in D^{r}(\mathcal{V}_{n}, P)$, thus we can choose $f = x^{k} = \pi^{k} \circ \phi$ to be the \emph{coordinate projection functions} with respect to the chart $(U_{P}, \phi)$. Applying \equref{transflaw} to the $x^{k}$ functions ($k = 1, 2, ..., n$) we get
\begin{small}
$$
\frac{\partial}{\partial x^{i '}} \Big|_{P}(x^{k}) = \lambda_{i '}^{j} \frac{\partial}{\partial x^{j}} \Big|_{P}(x^{k}) \Leftrightarrow
\Big( \frac{\partial x^{k}}{\partial x^{i '}} \Big)_{P} = \lambda_{i '}^{j} \Big( \frac{\partial x^{k}}{\partial x^{j}} \Big)_{P} = \lambda_{i '}^{j} \delta_{j}^{k} = \lambda_{i '}^{k} \Rightarrow
$$
\end{small}
\begin{equation} \eqlabel{lincoef}
\Rightarrow \lambda_{i '}^{k} = \Big( \frac{\partial x^{k}}{\partial x^{i '}} \Big)_{P} \in \mathbb{R}
\end{equation}
Combining \equref{transflaw}, \equref{lincoef} we arrive at the \emph{transformation law} (\emph{change of basis})
\begin{equation}  \eqlabel{trasflawfin}
\boxed{
\frac{\partial}{\partial x^{i '}} \Big|_{P} = \Big( \frac{\partial x^{j}}{\partial x^{i '}} \Big)_{P} \  \frac{\partial}{\partial x^{j}} \Big|_{P}
}
\end{equation}
between the two different natural bases, of the tangent space $T_{P}(\mathcal{V}_{n})$ at the point $P$, with respect to the two different local charts $(U_{P}, \phi)$, $ \ (V_{P}, \psi)$.
\begin{remark}
The reader should note at this point that the (real) numerical coefficients \equref{lincoef} of the linear transformation \equref{trasflawfin} (i.e. the change of basis law in the tangent space) are given by
\begin{equation} \eqlabel{Jacobofchangofbas}
\lambda_{i '}^{j} = \Big( \frac{\partial x^{j}}{\partial x^{i '}} \Big)_{P} = D_{i '} (x^{j} \circ \psi^{-1}) \Big|_{\psi(P)} = D_{i '} \big( \pi^{j} \circ ( \phi \circ \psi^{-1}) \big) \Big|_{\psi(P)}
\end{equation}
i.e. they are exactly the entries of the Jacobian matrix of the transition function from the local coordinate system $(V_{P}, \psi)$ to the local coordinate system $(U_{P}, \phi)$. \\
In the above, $D_{i '}$ denotes the usual partial derivative with respect to the $x^{i '}$ coordinate of $\mathbb{R}^{n}$. (\equref{Jacobofchangofbas} is obtained by applying \equref{differonmanif} or equivalently applying \leref{transferofdiffer})
\end{remark}

\textbf{\underline{Exercises:}}   \\
$\mathbf{1.}$ Show that the inverse transformation of \equref{trasflawfin} is given by $\frac{\partial}{\partial x^{k}} \big|_{P} = \big( \frac{\partial x^{i '}}{\partial x^{k}} \big)_{P} \  \frac{\partial}{\partial x^{i '}} \big|_{P}$
(all indices sum up to the dimension $n$ of the manifold)   \\
$\mathbf{2.}$ Choose any point on the sphere $\mathcal{S}^{2}$ belonging to two different local charts of the stereographic projection. Find the tangent space, two bases (with respect to the two different local charts), and the transformation laws (between these two bases)   \\
$\mathbf{3.}$ Prove that the abstract (algebraic) definition of the tangent vector presented here and the geometrical definition of the tangent vector \footnote{By geometrical definition we imply the definition of the tangent vector as an equivalence class of curves on the manifold (see the bibliography)} are equivalent. In other words establish an explicit isomorphism between the tangent (vector) space consisting of linear functionals satisfying the chain rule of differentiation \equref{abstrdeftangect} and the tangent (vector) space consisting of suitable equivalence classes of curves on the manifold.

%\subsubsection{Geometric definition}

\subsection{The tangent bundle and its structure as a manifold}

We show in the previous subsection, that in any point $P$ of a differentiable manifold $\mathcal{V}_{n}$ one can attach a collection $T_{P}(\mathcal{V}_{n})$ of tangent vectors and that this collection $T_{P}(\mathcal{V}_{n})$ has the structure of a real vector space. We determined ``natural'' bases for these vector spaces and laws of transformtion between these bases.

Our next step will be to consider the collection of all these vector spaces together. This naturally leads us to the notion described in the following definition
\begin{definition}[\textbf{Tangent Bundle} $\mathbf{T(\mathcal{V}_{n})}$ of a manifold $\mathcal{V}_{n}$]
Let $\mathcal{V}_{n}$ be an $n$-dimensional, $\mathcal{C}^{r}$-differentiable manifold. We define the tangent bundle $T(\mathcal{V}_{n})$ of the manifold to be
\begin{equation}
\mathbf{T(\mathcal{V}_{n}) = \bigcup_{P \in \mathcal{V}_{n}} T_{P}(\mathcal{V}_{n})}
\end{equation}
the union of all tangent spaces $T_{P}(\mathcal{V}_{n})$ at all points $P$ of the manifold
\end{definition}
We now have the following theorem which establishes the structure of a differentiable manifold on the tangent bundle
\begin{theorem}
Let $\mathcal{V}_{n}$ be an $n$-dimensional, $\mathcal{C}^{r}$-differentiable manifold. Its tangent bundle $T(\mathcal{V}_{n})$ is a $2n$-dimensional, $\mathcal{C}^{r}$-differentiable manifold.
\end{theorem}
%\begin{proof}
%\end{proof}

\subsection{(Tangent) Vector fields}

From now on and for the rest of these notes all manifolds will be considered to be $C^{\infty}$ without further mentioning it. We are going to present three different -but equivalent as we shall see- definitions of the notion of a smooth vector field on a manifold.

Let $\mathcal{V}_{n}$ be an $n$-dimensional, ($\mathcal{C}^{\infty}$)-differentiable manifold.
\begin{definition}[\textbf{(Smooth) vector field} on a manifold]  \delabel{def-svf1}
Any map of the form
\begin{equation} \eqlabel{svf1}
\begin{array}{c}
\mathbf{X: \mathcal{V}_{n} \rightarrow T(\mathcal{V}_{n})}  \\
    \\
\mathbf{P \mapsto X(P) \equiv X_{P} \in T_{P}(\mathcal{V}_{n})}
\end{array}
\end{equation}
will be called a \textbf{vector field} on the manifold $\mathcal{V}_{n}$.

According to \equref{lincombbasveclccoordsyst} we have that if $(U_{P}, \phi)$ is a local chart around $P$ (equiv.: a local coordinate system $(x^{1}, x^{2}, ..., x^{n})$ in some neighborhood of $P$) then
$$
X_{P} = \sum_{i=1}^{n} X_{P}^{i} e_{i} \equiv \sum_{i=1}^{n} X_{P}^{i} \frac{\partial}{\partial x^{i}} \Big|_{P}
$$
If the functions $X^{i} \equiv X(x^{i}): U_{P} \rightarrow \mathbb{R}$ whose values at $P \in U_{P}$ are $X_{P}^{i} \equiv X_{P}(x^{i}) \in \mathbb{R}$ are smooth functions then the vector field \equref{svf1} is called a \textbf{smooth vector field} $\mathbf{X}$ on the manifold $\mathcal{V}_{n}$.
\end{definition}

\begin{definition}[\textbf{(Smooth) vector field} on a manifold]  \delabel{def-svf2}
A \textbf{smooth vector field} $\mathbf{X}$ on $\mathcal{V}_{n}$ is a smooth assignment of a tangent vector $X_{P} \in T_{P}(\mathcal{V}_{n}$ for any $P \in \mathcal{V}_{n}$ where smooth is defined to mean that for any $f \in D^{\infty}(\mathcal{N})$ \footnote{see the notation introduced in \equref{algdiffunctonN}} the function
\begin{equation} \eqlabel{svf2}
\begin{array}{c}
\mathbf{Xf: \mathcal{V}_{n} \rightarrow \mathbb{R}} \\
   \\
\mathbf{P \mapsto (Xf)(P) = X_{P}(f)}
\end{array}
\end{equation}
is infinitely differentiable.
\end{definition}
\begin{remark}
Notice that an immediate consequence of the \deref{def-svf2} is that a vector field gives rise to a map
\begin{equation}
\begin{array}{ccccc}
D^{\infty}(\mathcal{N}) \rightarrow D^{\infty}(\mathcal{N}) & & given \ by & & f \mapsto Xf
\end{array}
\end{equation}
The image of $f \in D^{\infty}(\mathcal{N})$ under the above map is $Xf \in D^{\infty}(\mathcal{N})$ and is often known as the \emph{Lie derivative of the function $f$ along the vector field $X$} and is then denoted by \emph{$L_{X}f$} or \emph{$\pounds_{X}f$}
\end{remark}

\begin{definition}[\textbf{(Smooth) vector field} on a manifold]  \delabel{def-svf3}
A \textbf{smooth vector field} $\mathbf{X}$ on $\mathcal{V}_{n}$ is a smooth \footnote{in the sense of \deref{smooth}} map
\begin{equation} \eqlabel{svf3}
\begin{array}{c}
\mathbf{X: \mathcal{V}_{n} \rightarrow T(\mathcal{V}_{n})}  \\
    \\
\mathbf{P \mapsto X(P) \equiv X_{P} \in T_{P}(\mathcal{V}_{n})}
\end{array}
\end{equation}
from the $n$-dimensional manifold $\mathcal{V}_{n}$ to the $2n$-dimensional manifold $T(\mathcal{V}_{n})$.
\end{definition}

\textbf{\underline{Exercises:}}   \\
$\mathbf{4.}$ Prove the equivalence of the three previously stated definitions of the notion of smooth vector field. \\

%\subsection{Cotangent space, Cotangent bundle, and (cotangent) 1-form fields}

%\section{Differential of a function}

%\section{Pullbacks - Pushouts}

%\section{Implicit function theorem on a manifold}

%\section{Immersions - Imbeddings}

%\section{Inverse mapping theorem on a manifold}

%\section{Submanifolds}

%\section{Affine connections, Covariant Derivatives, Christoffle symbols}

%\section{Torsion (of a connection)}

%\section{Curvature (of the connection)}

%\section{Lie groups, Lie algebras and the exponential mapping}

%\section{Riemannian metrics and manifolds}

%\section{Introduction to curvature (scalar, Ricci, Gauss, ...)}

%\section{Gauss's Lemma, Bianchi identities}

\end{document}